\date{}
\documentclass[11pt]{article}
\usepackage{amsmath,amsthm,amsfonts,amssymb,amsopn,amscd}     
\usepackage{color}

\DeclareMathOperator{\Ker}{Ker}
\DeclareMathOperator{\Stab}{Stab}

\DeclareMathOperator{\Ind}{Ind}

\def\tilde{\widetilde}

\def\a{\alpha}

\def\setminus{\smallsetminus}

\def\A{{\cal A}}
\def\B{{\cal B}}

\def\C{{\cal C}}

\def\F{{\cal F}}
\def\cF{{\cal F}}

\def\R{{\cal R}}
\def\LL{{\mathcal L}}

\def\H{{\cal H}}
\def\K{{\cal K}}
\def\S{{\cal S}}

\def\PSL{{{\rm PSL}(2,\mathbb R)}}

\def\S2{S^{1(2)}}
\def\Poi{{\cal P}_+^\uparrow}
\def\tPoi{\tilde{\cal P}_+^\uparrow}
\newtheorem{theorem}{Theorem}[section]
\newtheorem{lemma}[theorem]{Lemma}

\newtheorem{corollary}[theorem]{Corollary}

\newtheorem{proposition}[theorem]{Proposition}

\theoremstyle{definition} 

\theoremstyle{remark} 

\def\setminus{\smallsetminus}

\def\PSL{PSU(1,1)}

\def\Z{{\mathbb Z}}

\def\RR{{\mathbb R}}
\def\CC{{\mathbb C}}

\def\SL2{{{\rm SL}(2,\RR)}}
\def\SLC{{{\rm SL}(2,\CC)}}
\def\PSL2{{{\rm PSL}(2,\RR)}}

\def\U1{{{\rm V}(1)}}
\def\SU2{{{\rm SV}(2)}}

\def\SO{{{\rm SO}}}
\def\SU{{{\rm SU}}}

\def\A{{\mathcal A}}
\def\B{{\mathcal B}}
\def\C{{\mathcal C}}

\def\F{{\mathcal F}}
\def\H{{\mathcal H}}

\def\K{{\mathcal K}}

\def\O{{\mathcal O}}

\def\cP{{\mathcal P}}
\def\R{{\mathbb R}}
\def\T{{\mathcal T}}

\def\W{{\mathcal W}}
\def\cW{{\mathcal W}}

\def\eins{{\mathbf 1}}

\title{\Huge{Where Infinite Spin Particles\\ are Localizable}}

\author{{\sc Roberto Longo\footnote{Supported in part by the ERC
      Advanced Grant 669240 QUEST ``Quantum Algebraic Structures and
      Models'', PRIN-MIUR, GNAMPA-INdAM, and Alexander von Humboldt
    foundation.}}
\\
Dipartimento di Matematica,
Universit\`a di Roma Tor Vergata,\\
Via della Ricerca Scientifica, 1, I-00133 Roma, Italy\\
E-mail: {\tt longo@mat.uniroma2.it}
\\
\phantom{X}\\
{\sc Vincenzo Morinelli\footnote{Supported in part by
PRIN-MIUR and GNAMPA-INdAM.}}
\\
Dipartimento di Matematica,
Universit\`a di Roma Tor Vergata,\\
Via della Ricerca Scientifica, 1, I-00133 Roma, Italy\\
E-mail: {\tt morinell@mat.uniroma2.it}
\\
\phantom{X}\\
{\sc Karl-Henning Rehren}
\\
Institut f\"ur Theoretische Physik, Universit\"at G\"ottingen,\\ 
37077 G\"ottingen, Germany\\
E-mail: {\tt rehren@theorie.physik.uni-goettingen.de}
}

\date{}
\begin{document} 

\maketitle

\begin{abstract}
Particles states transforming in one of the infinite spin
  representations of the Poincar\'e group 
(as classified by E. Wigner) 
  are consistent with
  fundamental physical principles, but local fields generating them
  from the vacuum state cannot exist. While it is known that
  infinite spin states localized in a spacelike cone are dense in the
  one-particle space, we show here that the subspace of states
  localized in any double cone is trivial. This implies that the
  free field theory associated with infinite spin has no observables
  localized in bounded regions. In an interacting theory, if the
  vacuum vector is cyclic for a double cone local algebra, then the theory
  does not contain infinite spin representations. 
We also prove
that if a Doplicher-Haag-Roberts representation (localized in a double
cone) of a local net is covariant under a unitary representation of the Poincar\'e group containing
infinite spin, then it has infinite statistics. 

These results hold under the natural assumption of the Bisognano-Wichmann property, and we give 
a counter-example (with continuous particle degeneracy) without this property where the conclusions fail.
Our results hold true in any spacetime dimension $s+1$ where infinite spin representations exist, namely $s\geq 2$.

\end{abstract}

\newpage

\section{Introduction}
\label{sec:intro}
The classical notion of particles as pointlike
objects is meaningless in Quantum Mechanics. Here the wave function
satisfies the Schr\"odinger equation and the Heisenberg uncertainty
relation prevents a sharp localization; increasing energy is needed
for better localization. We are going to discuss the intrinsic
particle localization properties, and show why infinite spin particles
exhibit an essential difference from finite spin particles in this respect.
\smallskip

\noindent
{\it Wigner particles and classification of Poincar\'e group representations.}
In Relativistic Quantum Mechanics, one better defines a particle through its symmetry, rather than localization, property. The Schr\"odinger equation is replaced by the Lorentz invariant Klein-Gordon equation and this point of view led to define a particle as an irreducible, positive energy, projective unitary representation of the Poincar\'e group $\Poi$, hence to an irreducible, positive energy, unitary representation of the double (universal) cover of $\Poi$. These are the ``minimal" Poincar\'e covariant objects, the building blocks of any more complete theory.

Within this point of view, E. Wigner \cite{W} obtained his famous
classification of the irreducible, positive energy, unitary
representations of the double cover of $\Poi$, which is isomorphic to
$\RR^4\rtimes\SLC$. 

We briefly recall that a unitary, positive energy representation $U$
is classified, up to unitary equivalence, by two parameters $m$ and
$s$. The \emph{mass} $m$ takes values in $[0,\infty)$ (the lower point
in the energy spectrum). If $m>0$, then the values of the \emph{spin}
$s$ are $0, \frac12, 1, \frac32, 2,\dots$ (the unitary representations of the cover of the rotation subgroup).

In the mass zero case, the representations fall in two distinct
classes according to the representations of the little group, which is
the double cover of $E(2)$, the Euclidean group of the plane. The representation with trivial $E(2)$-translations are representations of the (double) torus, labelled by the \emph{helicity}, a parameter $s$ that takes the place of the spin, $s = 0, \frac12, 1, \frac32, 2,\dots$.

The remaining massless representations correspond to infinite-dimensional, irreducible representations of the double cover of $E(2)$ and are labelled by a parameter $\kappa >0$ (the \emph{radius} of the circle that is the joint spectrum of the $E(2)$-translations) and a $\pm$ sign (Bose/Fermi alternative). They are called \emph{infinite spin} (or continuous spin) representations.

Infinite spin particles have so far not been
observed in nature, although they are compatible with all physical
first principles, and are usually disregarded without
further explanation. A result by Yngvason \cite{Y} shows
that they cannot appear in a Wightman theory \cite{Y}
since no Wightman fields (which have pointlike localization)
transforming under an infinite spin representation can exist. One of the main aims of this paper is to study
the peculiar localization property of these particles so as to explain
why they are not observable in finite space and time. They are
however localizable in certain unbounded spacetime
regions \cite{BGL} (cf.\ also \cite{IM}). Indeed, the authors of \cite{MSY} have constructed such fields $\Phi(x,e)$
which are localized along rays $x+\RR_+\cdot e$ where $e$ is a
spacelike direction.

We mention at this point that more general notions of particle are necessary to describe situations where, for example, infrared clouds are present, cf. \cite{BPS}; these will not be considered in the present paper.

The main body of this article deals with the issue of localization
of (one-particle) states. In Sect.\ \ref{sec:QFT}, we present some
consequences for the localization of algebras of observables (in the
sense of spacelike commutation relations). The one-particle results
directly pass to {\em free} fields by second quantization, and we shall discuss
general results in the interacting case.
\smallskip

\noindent
{\it Localized particle states.}
Given a particle, namely an irreducible, positive energy
representation $U$ of $\Poi$, what are the localized states of of $U$? 

If we restrict our attention to finite spin particles, the answer is 
well known in the Quantum Field Theory context, 
where one assumes the 
existence of a local free field transforming in a given
representation.
In the scalar case, for simplicity, the one-particle Hilbert space $\H$ can
be obtained by equipping the Schwartz function space $S(\RR^{s+1})$
  with a scalar product, given by the two-point function 
  $(f,g)=(\Phi(f)\Omega,\Phi(g)\Omega)$ of the field. Its Hilbert space closure
  $\H$ can be viewed as the space of positive energy solutions to the
  Klein-Gordon wave equation, and carries an irreducible
  representation $U$ of $\Poi$ with zero spin/helicity. The
  localization of one-particle states is given by the support of the
  Schwartz functions: by assigning to an open region  
$X$
the closed real linear subspace  
$H_\Phi(X)\subset\H$,
the closure of the space of
  real smooth functions with support in $X$, one obtains a local
  $U$-covariant net of standard subspaces of $\H$ (see below). The
  locality of the field, together with the identity $i\Im(f,g) = (\Omega,[\Phi(f)^*,\Phi(g)]\Omega)$, 
imply that two subspaces $H_\Phi(X)$ and
  $H_\Phi(Y)$ are symplectically orthogonal whenever $X$ and $Y$ are
  spacelike separated. 
  
In the sequel, we describe the procedure of {\em modular}
  localization, which 
intrinsically associates with a given representation the states
localized in a region $X$, without referring to a local
field. 
\smallskip

\noindent
{\it Terminology.}
A wedge region $W$ is a Poincar\'e transform
  of the standard wedge $W_0=\{x\in\RR^4: x_3>\vert x_0\vert\}$, and
  $\W$ is the set of all wedge regions. The standard one-parameter
  family of boosts preserving $W_0$ is called $\Lambda_{W_0}(t)$, and
  we put $\Lambda_W(t):=g\Lambda_{W_0}(t)g^{-1}$ if $W=g(W_0)$.  
  A double cone $O$ is
  the open intersection of a future and a backward light cone, and
  $\O$ is the set of all double cones. A spacelike cone is a
  region of the form $C=x+\bigcup_{t>0}t\cdot O$ where $x\in\RR^4$ and
  $O\in\O$ is a double cone spacelike to the point $0$, and $\C$ is 
  the set of all spacelike cones. Two regions $X$, $Y$ are spacelike 
  separated if every pair of points $(x,y)\in X\times Y$ is
  spacelike separated. The spacelike complement of a region $X$ is
  denoted by $X'$.
\smallskip

\noindent
{\it Standard subspaces.}
Let $\H$ be a Hilbert space. A \emph{standard subspace} $H$ of $\H$ is
a closed, real linear subspace which is cyclic ($H +iH$ is dense) and
separating ($H\cap iH = \{0\}$). Standard subspaces of the
one-particle space naturally appear in the above free field
construction: the standardness of $H_\Phi(O)$ is equivalent to the Reeh-Schlieder property that
the vacuum vector is cyclic and separating for the corresponding local
von Neumann algebras $\A(O)$ \cite{A}.

If $H$ is a standard subspace, the Tomita operator $S: \xi +
i\eta\mapsto \xi - i\eta$, $\xi,\eta\in H$, remembers $H$ as
$H=\Ker(S-1)$, and its polar decomposition $S=J\Delta^{1/2}$ gives the modular operator $\Delta$ and the modular conjugation $J$ that satisfy the one-particle version of the Tomita-Takesaki theorem: 
\begin{equation}
\Delta^{it}H = H \quad \forall\,t\in\RR\ , \qquad JH = H' \ ,
\end{equation}
where $H'$ is the symplectic complement of $H$.
\smallskip

\noindent 
{\it The Bisognano-Wichmann property.} Now, let $U$ be a positive energy representation of $\Poi$ on $\H$,
and $\W\ni W\longmapsto H(W)$ a net of standard subspaces on the wedge regions
of the Minkowski spacetime $\RR^4$, which is $U$-covariant: 
\[
U(g)H(W) = H(gW) \ .
\]
The Bisognano-Wichmann property \cite{BW} asserts that the modular group
of $H(W)$ is related to the boost transformations $\Lambda_W$
preserving $W$: 
\begin{equation}\label{BW}
\Delta_W^{it} = U\big(\Lambda_W(-2\pi t)\big)\ .
\end{equation}
If $H(C)$ is cyclic for all cones $C$, then $J_W$ acts geometrically
as a reflection around the edge of the wedge, so there exists an
anti-unitary PCT operator 
\begin{equation}\label{BW1}
\Theta \equiv U(R_W)J_{W}\ ,
\end{equation}
where $R_W$ is the spatial $\pi$-rotation mapping $W$ onto $W$. 

In Quantum Field Theory,
the Bisognano-Wichmann property pertains to the standard subspaces
$\overline{\A(W)_{\rm s.a.}\Omega}$ where $\A(W)_{\rm s.a.}$ is the
selfadjoint part of the von Neumann algebra of local observables in a
wedge. It was established model-independently for large classes of
quantum field theories, cf.\ Sect.\ \ref{comm:BW}. Because the
modular group is characterized by the KMS property, its physical
meaning is that the vacuum state is a KMS state for the boost
subgroup, when restricted to the algebra of a Rindler wedge; in other words, the restriction of the vacuum state is a thermal state for the geodesic observer on the Rindler spacetime.
By this feature it is closely
related to the Hawking-Unruh effect. We therefore believe the
Bisognano-Wichmann property to be of a most fundamental character, and
refer to the final comment \ref{comm:BW} and to \cite{H} for a discussion of this important point.
\smallskip

\noindent
{\it Modular localization.} 
The paper \cite{BGL} provided a canonical construction of a local net
$H_U$ of standard subspaces on the wedge regions of the Minkowski
spacetime $\RR^4$ associated with any unitary, positive energy,
representation $U$ of the Poincar\'e group (with anti-unitary PCT
operator $\Theta$). One \emph{defines} $\Delta_W$ and $J_W$ by the
equations (\ref{BW}), (\ref{BW1}), then sets $S_W \equiv J_W\Delta_W^{1/2}$ and
\begin{equation}\label{HUO}
H_U(W)\equiv \big\{\xi\in\H: S_W\xi =\xi\big\}\ .
\end{equation}
Isotony of the assignment $W\longmapsto H_U(W)$ 
  (i.e., $H_U(W_1)\subset H_U(W_2)$ whenever $W_1\subset W_2$) follows
  from positivity of the energy. Moreover $H_U$ is local (or twisted-local if we consider representations of the cover of the Poincar\'e group), indeed $H_U$ is wedge dual: $H_U(W')= H_U(W)'$.
  
This construction is
  intrinsic, depending only on the representation $U$ without
  reference to a quantum field.
By construction, $H_U$ satisfies the Bisognano-Wichmann property.

Notice that any net $W\longmapsto H(W)$ on wedges defines closed, real
linear subspaces associated with any region $X$ that is contained in
some wedge:
\begin{equation}\label{HO}
H(X) \equiv \bigcap_{\W\ni W\supset X}H(W) \ .
\end{equation}
Obviously, these definitions respect isotony ($H(X_1)\subset H(X_2)$
whenever $X_1\subset X_2$), and locality.

If $H(O)$ is cyclic
for double cones $O$, then $H(W)$ defined by additivity from the
double cones coincides with the original $H(W)$ (assuming the Bisognano-Wichmann property).

A general result \cite{BGL} shows furthermore that $H_U(C)$ defined as in
\eqref{HO} from the canonical net \eqref{HUO} is standard for spacelike
cones $C\in\C$, for every representation $U$.   

If $U$ is a representation with finite spin/helicity, 
then the modular
localization subspace $H_U(X)$ as in \eqref{HO} agrees with the standard
subspace $H_\Phi(O)$ defined by the free field one-particle
construction recalled above,
therefore $H_U(O)$ is standard for any
double cone $O$, and in this case we explicitly see how the space
$H(X)$ of particle states localized in a bounded region $X$ is cyclic. 

We should also comment that  the paper \cite{BGL} deals with the bosonic case (true representation of $\Poi$), however
the fermionic case can be treated analogously with usual modifications (and quantization on the anti-symmetric Fock space).
\smallskip

\noindent
{\it Infinite spin particles cannot be localized in bounded regions.}
As recalled, in Wigner's classification of unitary, positive energy,
irreducible representations of the Poincar\'e group \cite{W}, massless
representations fit in two classes, the ones with finite spin
(helicity) and the ones with infinite spin, according to the
representations of the ``little group'', the double cover $\tilde
E(2)$ of the Euclidean group of the plane $E(2)$. 

Let $U$ be a massless representation with infinite spin; the space
$H_U(C)$ was shown to be standard (cyclic) for spacelike cones but it remained open whether there are non-zero vectors localized in bounded regions \cite{GL}. Generalized (stringlike) Wightman fields associated with $U$ were later constructed \cite{MSY}, but the above localization problem remained unsettled.

We shall show here that $H_U(X)$ is trivial if $X$ is bounded, say
$X=O$ a double cone, namely
\[
H_U(O) \equiv \bigcap_{\W\ni W\supset O}H_U(W) = \{0\} \ .
\]
\smallskip

\noindent
{\it Quantum Field Theory, I.} 
An immediate consequence is that the free field net $\A$ of local von Neumann algebras
associated with a representation $U$ of $\Poi$ with infinite spin is
well defined, the vacuum vector is cyclic for $\A(C)$ if $C$ is a
spacelike cone, but $\A(O)=\CC\cdot\eins$ if $O$ is a double cone: there is no non-trivial observable localized in a bounded spacetime region.

It also follows that there are no compactly localized observables on
the same Hilbert space that are relatively local w.r.t.\ the infinite
spin free field net. The absence of such observables was recently also
demonstrated within an explicit field theoretic ansatz \cite{K}.


An important more general corollary is that, if $\B$ is any (Fermi-)local net of
von Neumann algebras on a Hilbert space, covariant under a unitary
positive energy representation $U$ of the Poincar\'e
group, with the vacuum vector being cyclic (Reeh-Schlieder property)
for double cone algebras, then no infinite spin representation can appear in the irreducible direct integral decomposition of $U$ (up to measure zero), provided that $\B$ satisfies the fundamental Bisognano-Wichmann property \cite{BW}. 

This shows why infinite spin particles do not appear in a theory of
local observables. They are however compatible with stringlike
localization. At this point it is worth mentioning that localization
in spacelike cones is natural in Quantum Field Theory, indeed massive
charges may always be localized in spacelike cones
\cite{BF}. Low-dimensional non-trivial models with trivial local
algebras are exhibited in \cite{LL}.

\smallskip

\noindent
{\it Strategy of proof.} 
Let $U$ be a unitary, massless irreducible representation of $\Poi$.
The starting point is the observation that $U$ is dilation covariant if and only if it has finite spin. 
Assuming $H_U(O)$ to be standard for double cones $O$, we infer by the
Huygens principle that 
$H_U(V_+)$ is standard, where $V_+$ is the forward light cone. By standard subspace analysis, in particular by using an analogue of Borchers' theorem \cite{B,L}, $\Delta_{H_U(V_+)}$ has dilation commutation relations with $U$. So $U$ must have finite spin.
\smallskip

\noindent
{\it Extensions of results.} 
Our results hold in any space dimension $s\geq 2$. As is known, if $s$
is even the Huygens principle doesn't hold and we need to work with a
corresponding property of the wave equation that we haven't found in
the literature. The case $s=2$ is peculiar as the infinite spin
representations are not ``infinite'', namely they are associated with
one-dimensional representations of the little group. The Fermi case,
namely representations of a cover of $\Poi$, is also studied. We treat
the case $s=3$ (the physical Minkowski spacetime) in detail and add a further section with the necessary analysis in different spacetime dimensions.
\smallskip

\noindent
{\it Quantum Field Theory, II.} 
For interacting theories satisfying the
Bisognano-Wich\-mann property, we show that the subspace
  $\overline{\B(O)\Omega}$ (independent of the double cone $O$) cannot carry an
  infinite spin representation. Thus, if the theory possesses infinite
spin particles, then the vacuum vector cannot be cyclic for $\B(O)$,
where $\B$ is the local net of von Neumann algebras describing our theory,
i.e., the infinite spin particles states cannot be generated from the
vacuum by operations in bounded spacetime regions. 

Indeed, no infinite spin particle state can be obtained by adding to $\B$ a finite charge 
localized in a bounded spacetime region. In other words, no infinite spin representation 
can appear in the irreducible disintegration of the covariance unitary 
representation of DHR sectors of $\B$ (with finite statistics).

We emphasize that the Bisognano-Wichmann property is essential in the
argument, by providing a counter-example without this
property, in which free infinite spin particles exist with cyclic double cone
algebras. 
 
Thus, at least one (artificial) way to
accomodate New Physics involving observable infinite spin particles
would consist in relaxing the Bisognano-Wichmann property -- in spite of its
very fundamental nature. More interesting is the picture (Sect.\ \ref{fieldstr}) that we obtain
when we start with a (compactly) local observable net;
we have a field algebra net that generates a non-trivial but non-cyclic subspace,
an \emph{interacting} theory with infinite spin particles; this
structure exactly complies with the picture envisaged in \cite{Sch}.

\section{Standard subspaces}
We begin by recalling some definitions and results on standard
subspaces and their modular structures. Further details can be found
in \cite{L,LN}. 

A linear, real, closed subspace $H$ of a complex Hilbert space $\H$ is called \emph{cyclic} if 
$H+iH$ is dense in $\H$, \emph{separating} if $H\cap iH=\{0\}$ and 
\emph{standard} if it is cyclic and separating.

Given a standard subspace $H$ one defines the Tomita operator $S_H$, the closed, anti-linear involution
with domain $H + iH$, given by $S_H : \xi + i\eta \mapsto \xi + i\eta$, $\xi,\eta\in H$. The polar decomposition $S_H = J_H\Delta_H^{1/2}$ defines the positive selfadjoint \emph{modular operator} $\Delta_H$ and the anti-unitary
\emph{modular conjugation} $J_H$. $\Delta_H$ is invertible and $J_H\Delta_H J_H=\Delta_H^{-1}$.

Pairs $(J,\Delta)$, where $J$ is an anti-unitary involution and $\Delta$ a selfadjoint positive invertible operator s.t.\ $J\Delta J=\Delta^{-1}$ are in 1-1 correspondence with closed, anti-linear, densely defined involutions $S=J\Delta^{1/2}$ and in 1-1 correspondence with standard subspaces $H=\Ker(S - 1)$.

If $H$ is a closed, real linear subspace of $\H$, the \emph{symplectic complement} of $H$ is defined by
\[
H' \equiv \{\xi\in\H\ ;\ \Im(\xi,\eta)=0\ \forall \eta\in H\} = (iH)^{\bot_\R}\ ,
\]
where $\bot_\R$ denotes the orthogonal in $\H$ viewed as a real Hilbert space with respect to the real part of the scalar product.
$H'$ is a closed, real linear subspace of $\H$ and $H = H''$. 

$H$ is cyclic (separating) iff $H'$ is separating (cyclic), thus $H$ is standard iff $H'$ is standard and we have
\[
S_{H'} = S^*_H \ .
\]
The fundamental properties of the modular operator and conjugation are
\begin{equation*}
\Delta_H^{it}H = H, \quad J_H H = H' \ ,\qquad  t\in\RR\ ,
\end{equation*}
and $t\mapsto \Delta_H^{it}$ is called the one-parameter unitary
\emph{modular group} of $H$ (cf.\ \cite{T,RV}). 

Let $H$ be a real linear subspace of $\H$ and $V$  a one-parameter
group of unitaries on $\H$ such that $V(t)H = H,\ t\in\RR$. $V$
satisfies the \emph{KMS condition} with inverse temperature $\beta>0$ on $H$ if, for every given $\xi,\eta\in H$, there exists a function $F$, analytic on the strip $\{z\in\CC: 0<\Im z<1\}$, bounded and continuous on its closure,  such that:
\begin{gather*}
F(t)=\langle \eta, V (t)\xi\rangle\ ,\qquad t\in\R\ , \\
F(t+i\beta)=\langle V(t)\xi,\eta\rangle\ ,\qquad t\in\R\ .
\end{gather*}
Since the uniform limit of holomorphic functions is holomorphic, it follows that if the KMS condition holds on $H$, then it holds on the closure $\overline{H}$ of $H$.
\begin{lemma}{\rm \cite{L,LN}.}
If $H\subset\H$ is a standard subspace, then $t\mapsto\Delta_H^{-it}$ satisfies the KMS condition at inverse temperature 1.

Conversely, if $H$ is a closed, real linear, cyclic subspace of $\H$
and $V$ a one-parameter unitary group on $\H$ with $V(t)H = H, \ t\in\RR$, satisfying the KMS condition on $H$ at inverse temperature 1, then $H$ is standard and $V(t)=\Delta^{-it}_H$.
\end{lemma}
\noindent
The following lemma is a consequence of the KMS condition for the modular group.
\begin{lemma}\label{inc}{\rm \cite{L,LN}.}
Let $H\subset \H$ be a standard subspace, and $K\subset H$ a closed, real linear subspace of $K$. 

If $\Delta_H^{it}K=K$, $\forall t\in\R$, then $K$ is a standard subspace of $\K\equiv \overline{K+iK}$ and $\Delta_H|_K$ is the modular operator of $K$ on $\K$. 
If moreover $K$ is a cyclic subspace of $\H$, then $H=K$.
\end{lemma}
\noindent
We shall also need the following basic lemma.
\begin{lemma}\label{inc2}{\rm \cite{L,LN}.}
Let $H\subset \H$ be a standard subspace, and $U$ a unitary on $\H$ such that $UH=H$. Then $U$ commutes with $\Delta_H$ and $J_H$.
\end{lemma}
\noindent
The following is the one-particle analogue of Borchers' theorem \cite{B}.
\begin{theorem}{\rm \cite{L,LN}.}\label{Borch}
Let $H\subset\H$ be a standard subspace, and $U$ a one-parameter unitary group on $\H$ with positive generator, such that $U(t)H\subset H$, $t\geq 0$. Then $\Delta_H^{is}U(t)\Delta_H^{-is}= U(e^{-2\pi s}t)$.
\end{theorem}

\noindent
We now want to study the tensor product of standard subspaces.
Let $H$ and $K$ be standard subspaces  of the Hilbert spaces $\H$ and
$\K$ respectively, and $S_H,\,S_K$ the associated Tomita operators. Then
$S\equiv S_H\otimes S_K$
is a closed, densely defined anti-linear involution. Therefore $S=S_M$ where 
$M \equiv \big\{\xi\in {\rm Dom}(S) : S\xi=\xi\big\}$
is a standard subspace of $\H\otimes\K$. 

We define the \emph{tensor product} of $H$ and $K$ by $H\otimes K\equiv M$; in other words $H\otimes K$ is defined through the formula
\[
S_{H\otimes K}\equiv S_H\otimes S_K\ .
\]
\begin{proposition}\label{com}
If $H$ and $K$ are standard subspaces of $\H$ and $\K$ respectively, we have
\[
(H\otimes K)' = H'\otimes K'\ .
\]
\end{proposition}
\begin{proof}
Immediate from the equality 
\[
S_{(H\otimes K)'} = S_{H\otimes K}^* = \big(S_H\otimes S_K\big)^* = 
S_H^*\otimes S_K^* = S_{H'\otimes K'}\ .
\]
\end{proof}
With $H$, $K$ real linear subspaces of $\H$ and $\K$ respectively we denote by
$H\odot K$ the real linear span of $\{\xi\otimes \eta:\xi\in H,\,\eta\in K\}$.
\begin{proposition}\label{lem:b}
Let $H$ and $K$ be standard subspaces of $\H$ and $\K$. We have:
\[
H\otimes K=\overline{H\odot K}\ .
\]
\end{proposition}
\begin{proof}
$\overline{H\odot K}$ is cyclic since $H\odot K+iH\odot K = (H+iH)\odot(K+iK)$, which is dense in $\H\otimes\K$.
Clearly $H\odot K$ is in the domain of $S_H\otimes S_K=S_{H\otimes K}$, thus
$H\odot K\subset H\otimes K$.
Now
$\Delta^{it}_{H\otimes K}=\Delta^{it}_H\otimes \Delta^{it}_K$ leaves globally invariant $H\odot K$, hence $\overline{H\odot K}$.
By Lemma \ref{inc} we conclude that $\overline{H\odot K}$ is equal to $H\otimes K$.
\end{proof}
By Prop.\ \ref{lem:b}, we may equivalently define the tensor product of the closed, real linear subspaces $H$ and $K$ of $\H$ and $\K$ by $H\otimes K \equiv \overline{H\odot K}$.

Given a family of real linear subspaces $H_a$ of $\H$, we shall denote by $\sum_a H_a$ the real linear span of the $H_a$'s.

\begin{lemma}\label{lem:1}
Let $\{H_a\}$ be a family of closed, real linear subspaces of $\H$. Then
\[
\Big(\bigcap_a H_a\Big)' = \overline{\sum_a H'_a }\ .
\]
\end{lemma}
\begin{proof}
We have
\[
\Big(\bigcap_a H_a\Big)'= \Big(i\bigcap_a H_a\Big)^{\bot_\R} = \Big(\bigcap_a iH_a\Big)^{\bot_\R} = \overline{\sum_a (iH_a)^{\bot_\R}} = \overline{\sum_a H'_a }\ .
\]
\end{proof}
\begin{lemma}\label{lem:2}
Let $\{H_a\}$ and $\{K_b\}$ be families of standard subspaces of $\H$ and $\K$ respectively,  and suppose both the intersections $H \equiv \bigcap_a H_a$ and 
$K\equiv \bigcap_b K_b$ to be cyclic. We have:
\[
H\otimes K = \bigcap_{a,b}\big(H_a\otimes K_b\big)\ .
\]
\end{lemma}
\begin{proof}
By Lemma \ref{lem:1} we have to show that 
\[
(H\otimes K)' = \overline{\sum_{a,b}\big(H_a\otimes K_b\big)'} \ .
\]
By Prop.\ \ref{com}, we have indeed:
\[
\big(H\otimes K\big)'=H'\otimes K'
= \overline{\sum_a H'_a }\otimes  \overline{\sum_b K'_b } = \overline{\sum_{a,b} H'_a \otimes  K'_b }=
\overline{\sum_{a,b}\big(H_a\otimes K_b\big)'}\ .
\]
\end{proof}
\section{Massless representations of the Poincar\'e group}\label{masslessrep}
For the benefit of the reader, we first deal within the case of the
four-dimensional spacetime, later extending our results to different dimensions.

If $G$ is a locally compact group, $H\subset G$ a closed
 subgroup, and $V$ a unitary representation of $H$, we denote by $\Ind_{H\uparrow G} V$ the unitary representation of $G$ induced by $V$.

The Poincar\'e group $\Poi$ is the semi-direct product
$\RR^4\rtimes \LL_+^\uparrow$ of the proper orthochronous Lorentz group $\LL_+^\uparrow$ 
and the translation group $\RR^4$, where $\LL_+^\uparrow$ acts naturally on $\RR^4$.

The universal cover $\tilde\LL_+^\uparrow$ of $\LL_+^\uparrow$ is a double cover, isomorphic to $\SLC$. Accordingly, the universal cover $\tPoi$ of $\Poi$ is isomorphic to $\RR^4\rtimes\SLC$. 

One can choose the covering map $\sigma: \SLC\to\LL_+^\uparrow$,
so that $\sigma$ maps the one-parameter subgroup $\a$ 
\begin{equation}\label{alpha}
 \alpha(t) = \left(\begin{array}{cc}e^{t/2} &0\\0&e^{-t/2}
\end{array}\right) ,\ \ t\in\RR\ ,
\end{equation}
to the one-parameter group of boosts in the $x_3$-direction,
and $\sigma$ restricts to the usual covering map $\SU(2)\to \SO(3)$. Explicitly, one identifies a vector $x=(x_0,x_1,x_2,x_3)\in\RR^4$ with the matrix $X_x=\bigl( \begin{smallmatrix}
x_0 + x_3 &x_1 - ix_2\\ x_1 + ix_2& x_0 - x_3
\end{smallmatrix} \bigr)$ and defines the Lorentz transformation
$\sigma(A)\in \LL_+^\uparrow$ acting on $x$ through $X_{\sigma(A)x} = AX_xA^*$,
$A\in\SLC$, see \cite{SW}. 
 
The translation group $\RR^4$ is thus also a normal subgroup of
$\tPoi$. According to the Mackey machine (see \cite{Z}),
if $U$ is an irreducible unitary representation of $\tPoi$, then $U$ is induced by an irreducible unitary representation $U_0$ of $\overline{\Stab}_p$:
\begin{equation}\label{UU}
U = \Ind_{\overline{\Stab}_p\uparrow \tPoi}U_0 \ ;
\end{equation}
here the momentum $p\in\RR^4$ is a point in the dual group of the
translations (i.e., a character), $\overline{\Stab}_p$ is the
stabilizer of $p$ for the action of $\tPoi$ on the characters given by
the adjoint action on their arguments, and $U_0|_{\RR^4}$ is the one-dimensional representation $p$.

Notice that $\tilde\LL_+^\uparrow$
acts naturally on $\RR^4$ and $\RR^4$ acts trivially on
itself, so one has
\[
\overline{\Stab}_p =  \RR^4 \rtimes \Stab_p \ ,
\]
where $\Stab_p\subset \tilde\LL_+^\uparrow$ is the stabiliser of
$p$ in $\tilde\LL_+^\uparrow$ acting naturally on $\RR^4$ (the
\emph{little group}). Points $p$ in the same $\LL_+^\uparrow$-orbit
give rise to equivalent representations. 

We are interested in a positive energy, massless representation $U$,
thus $p\in\partial V_+$ the boundary of the forward light cone. We assume $U$ is not the identity, thus $p\neq 0$ and we shall choose and fix $p=q$ with
\[
q\equiv(1,0,0,1)\in \partial V_+
\]  
($\partial V_+\setminus\{0\}$ is a $\LL_+^\uparrow$-orbit).

Then $\Stab_q$, the little group of $(1,0,0,1)$, is isomorphic to $\tilde E(2)$, the double cover of the Euclidean group of the plane $E(2)$:
\begin{equation}\label{E2}
\Stab_q =\left\{ \left(\begin{array}{cc} u&z\\
0& \bar u
\end{array}\right):u, z\in\CC, \, |u| = 1\right\} \ .
\end{equation}

The irreducible representation $U_0$ of $\overline{\Stab}_p$ in \eqref{UU} has the form
\begin{equation}\label{U0}
U_0(g,x) = V(g)q(x)\ \!,\ g\in\Stab_q\ \!,\ x\in \RR^4\, ,
\end{equation}
where $V$ is an irreducible representation of $\Stab_q = \tilde E(2)$ and $q$ is the character of $\RR^4$.

Now $\tilde E(2)$ is the semi-direct product $\RR^2 \rtimes \mathbb T$ and an irreducible representation $V$ of $\tilde E(2)$ fits in one of the following two classes:
\begin{itemize}\itemsep0mm
\item[$(a)$] The restriction of $V$ to $\RR^2$ is trivial;
\item[$(b)$] The restriction of $V$ to $\RR^2$ is non-trivial.
\end{itemize}
Irreducible representations of $\tilde E(2)$ in class $(a)$ are thus
labelled by the integers, the dual of $\mathbb T$, while irreducible
representations in class $(b)$ are labelled by $\kappa >0$, the radius
of a circle in $\RR^2$, the joint spectrum of the $\tilde E(2)$-translations. 

We say in case $(a)$ that $U$ has \emph{finite spin} (or finite helicity); in case $(b)$ that $U$ has \emph{infinite spin}.
Therefore an irreducible, infinite spin representation $U$ of $\tPoi$ has the form
\begin{equation}\label{infspin}
U_{\kappa,\varepsilon} =\Ind_{\overline{\Stab}_q\uparrow\tPoi}\bar V_{\kappa,\varepsilon}
\end{equation}
where $\bar V_{\kappa,\varepsilon}$ is given by \eqref{U0}:
\[
\bar V_{\kappa,\varepsilon}(g,x) = V_{\kappa,\varepsilon}(g)q(x)\ \!,\ g\in \tilde E(2)\ \!,\ x\in \RR^4\, ,
\]
with $V=V_{\kappa,\varepsilon}$ is the representation of $\tilde E(2)$ in which
the spectrum of the translations is the circle of radius $\kappa>0$,
and the rotation by $2\pi$ is represented by $+\eins$ (bosonic case,
$\varepsilon=0$) resp.\ by $-\eins$ (fermionic case, $\varepsilon=\frac12$); so
infinite spin representations are labelled by $\kappa>0$ and $\varepsilon=0,\frac12$.
We shall denote by $\tau(z)$, $z\in\mathbb C$, the element of $\tilde E(2)\subset \SLC$ given by 
\[
\tau(z)=\left(\begin{array}{cc}1&z\\0&1\end{array}\right)\ ;
\]
the two translation one-parameter subgroups of $\tilde E(2)$ are
$\RR\ni x\mapsto\tau(x)$, and $\RR\ni y\mapsto\tau(iy)$ and we have the commutation relations
\begin{equation}\label{resc}
\alpha(t)\tau(z)\alpha(t)^{-1}=\tau(e^{t}z)\ .
\end{equation}
\section{Infinite spin representations are not dilation covariant}
As is known, an irreducible, massless finite helicity unitary representation extends, on the same Hilbert space, to a representation of the group of transformations of the Minkowski spacetime generated by $\tPoi$ and dilations (indeed to a unitary representation of the conformal group). We show here that irreducible infinite spin representations are not
dilation covariant in this sense. We suppress the Bose/Fermi label
$\varepsilon$ which is irrelevant for the issue at hand. 
\begin{lemma}\label{ind}
Let $G$ be a locally compact group, $H\subset G$ a closed subgroup and $\beta$ an automorphism of $G$ such that $\beta(H)=H$. If $V$ is a unitary representation of $H$ and
$U \equiv \Ind_{H\uparrow G} V$, then
\[
U\cdot \beta =  \Ind_{H\uparrow G} V\cdot\beta_0
\]
where $\beta_0 \equiv\beta |_{H}$.
\end{lemma}
\begin{proof}
The lemma follows by the unicity of the induced representation, a
consequence of the unicity of the measure class of a quasi-invariant
Borel measure on $H\backslash G$.  
\end{proof}
\begin{corollary}\label{cor}
Let $U_{\kappa}=\Ind_{\RR^4\rtimes  \tilde E(2)\uparrow \tPoi}\bar V_\kappa$ be an infinite spin, irreducible unitary representation of $\tPoi$, and $\beta$ an automorphism of $\tPoi$ preserving the element $q$ of (the dual of) the translation subgroup. Then $\beta(\Stab_q) = \Stab_q$ and
\[
U_\kappa\cdot\beta = U_{\kappa_\beta} \left(\equiv \Ind_{\RR^4\rtimes  \tilde E(2)\uparrow \tPoi}\bar V_{\kappa_\beta}\right) \ , 
\]
where $\kappa_\beta$ is given by $V_{\kappa_\beta} = V_\kappa\cdot \beta_0$ with $\beta_0$ the automorphism of $\tilde E(2)$ given by $\beta_0 = \beta|_{\Stab_q}$
\end{corollary}
\begin{proof}
This follows from Lemma \ref{ind}.
\end{proof}
We shall say that a unitary representation $U$ of $\tPoi$ on the Hilbert space $\H$ is \emph{dilation covariant} if $U$ extends to a unitary representation on $\H$ of the group generated by $\tPoi$ and dilations. Namely there exists a one-parameter unitary group $D(t)$ on $\H$ such that $D$ commutes with $U|_{\tilde\LL_+}$ and 
\[
D(t)U(x)D(-t) = U(e^t x)\ ,
\]
for $x$ in the translation group $\RR^4$.
\begin{proposition}\label{dil-inv}
Let $U$ be an irreducible, positive energy, unitary representation of $\tPoi$. Then $U$ is dilation covariant iff $U$ is massless with finite spin.
\end{proposition}
\begin{proof}
Let $\delta_t$ the the automorphism of $\tPoi$ given by $\delta_t(g)= g$ if $g\in\tilde\LL_+$ and $\delta_t(p) = e^t p$ if $p\in\RR^4$. We want to show that $U$ is inequivalent to $U\cdot \delta_t$, $t\neq 0$, if $U$ is irreducible with infinite spin.

Let then $U=U_\kappa$ be given by \eqref{UU}, namely $U_\kappa = \Ind_{\RR^4\rtimes  \tilde E(2)\uparrow \tPoi}\tilde V_{\kappa}$. We shall show that
\[
U_\kappa\cdot\delta_t = U_{e^{-t}\kappa} \ .
\]
This will prove the Proposition because $U_\kappa$ and $U_{\kappa'}$ are inequivalent if $\kappa\neq\kappa'$.

Now let $\a_t$ be the lift to $\tPoi$ of the inner one-parameter
automorphism group of $\Poi$ implemented by the boost in
$3$-direction, namely $\a$ is given by eq.\ \eqref{alpha}. Then 
\[
\a_t(q) = \delta_t(q) = (e^t , 0 , 0, e^t)\ ,
\]
where $q = (1,0,0,1)$ as above. Thus the automorphisms
\begin{equation}\label{beta}
\beta_t\equiv \a_{-t}\cdot\delta_t
\end{equation}
fix $q$. Since $\a_{-t}$ is inner, we have $U_\kappa\cdot\a_{-t}= U_\kappa$, thus
\[
U_\kappa\cdot\delta_t = U_\kappa\cdot\a_{-t}\cdot \delta_t  = U_\kappa\cdot\beta_t \ .
\]
We now apply Corollary \ref{cor} and see that
\[
U_\kappa\cdot\delta_t  = U_{\kappa'}
\]
where $\kappa'$ is given by
\[
V_{\kappa'} = V_\kappa\cdot \beta_t |_{\tilde E(2)} =  
V_\kappa\cdot \a_{-t }|_{\tilde E(2)} \ ,
\]
thus $\kappa' = e^{-t}\kappa$ by the commutation relation \eqref{resc} 
\cite[Lemma 4]{Y}.
\end{proof}
\section{Double cone localization implies dilation covariance}
\label{ax}
Let $U$ be a unitary, positive energy representation of the cover $\tPoi$ of the Poincar\'e group on a Hilbert space $\H$. 

A $U$-covariant \emph{net of standard subspaces} $\H$ on the set
$\W$ of wedge regions of the Minkowski spacetime is a map
\[
H: \W\ni W \longmapsto H(W)\subset\H
\]
that associates a closed real linear subspace $H(W)$ with each $W\in\cW$, satisfying:
\begin{enumerate}
\item {\it Isotony}: if $W_1\subset W_2$ then $H(W_1)\subset H(W_2)$;
\item {\it Poincar\'e covariance}:  $U(g)H(W)=H(gW)$,  $g\in \tPoi$;
\item {\it Reeh-Schlieder property}: $H(W)$ is cyclic $\forall \ W\in\cW$;
\item {\it Bisognano-Wichmann property}: 
\[
\Delta^{it}_{H(W)}=U\big(\Lambda_W(-2\pi t)\big),\qquad \forall\
W\in\W \ ;
\]
\item {\it Twisted locality}: 
For every wedge $W\in\cW$ we have
 \[
 Z H(W')\subset H(W)'
 \]
with $Z$ unitary, $\displaystyle{Z=\frac{\eins+i\Gamma}{1+i}}$.
\end{enumerate}
Due to twisted locality, each $H(W)$ is indeed a standard subspace, so the modular operators in Property 4 are defined.

Here $\Gamma \equiv U(2\pi)$, the unitary corresponding to a $2\pi$ spatial rotation in the representation $U$, namely $\Gamma$ is the image under $U$ of the non-trivial element in the centre of $\tilde\LL_+^\uparrow$.
Clearly $\Gamma$, hence $Z$, commutes with $U$.

Notice that if $U$ is bosonic ($\Gamma = \eins$), then $Z=\eins$, and
twisted locality is locality. If $U$ is fermionic ($\Gamma = - \eins$),  then $Z=-i$ and $H(W')\subset iH(W)$.

Lemma \ref{inc} then implies {\it twisted duality} for wedges:
\[
H(W')= ZH(W)'\ .
\]
Starting with a $U$-covariant net $H$ on $\W$ as above, one gets a net
of closed, real linear subspaces on double cones $O$
defined by
\begin{equation}\label{HO2}
 H(O)\equiv\bigcap_{\W\ni W\supset O}H(W) \ .
\end{equation}
Note that $H(O)$ is not necessarily cyclic. If $H(O)$ is cyclic, then
\[
H(W) = \overline{\sum_{O\subset W} H(O)}
\]
by Lemma \ref{inc}.

The following proposition is proved in \cite{BGL}, $(ii)\Rightarrow (i)$, and in \cite{GL}, $(i)\Rightarrow (ii)$, for nets of von Neumann algebras; yet the same argument gives a proof in the standard subspace setting.
\begin{proposition}{\rm \cite{BGL,GL}.}
\label{prop:ciclcone} 
Let $H$ be a $U$-covariant net of standard subspaces of $\H$ as above (properties 1--5). The following are equivalent:
\begin{itemize}\itemsep0mm
\item[$(i)$] $H(C)\equiv\bigcap_{\W\ni W\supset C}H(W)$ is cyclic for all spacelike cones $C$;
\item[$(ii)$] $U$ extends to an (anti-)unitary representation $\hat U$ of $\tilde\cP_+$ on $\H$ and $H$ is the canonical net $H_{\hat U}$ associated with $\hat U$ (eq.\ \eqref{HUO}). 
\end{itemize}
\end{proposition}
\noindent
Thus (in even spacetime dimension), with the above cone cyclicity assumption, there is an anti-unitary PCT operator.

The following proposition ensures a variant of the Reeh-Schlieder
property. If $O, \tilde O$ are double cones, we write $O\Subset\tilde
O$ if the closure of $O$ is contained in the interior of $\tilde O$.
\begin{proposition}\label{prop:cicl} 
Let $H(O)$ be defined as above in \eqref{HO2}, with $U$ irreducible.
If $H(O)\neq\{0\}$ for some double cone $O$, then $H(\tilde O)$ is
cyclic for every double cone $O\Subset\tilde O$.
\end{proposition}
\begin{proof}
Let $O\Subset\tilde O$ be double cones with $H(O)\neq\{0\}$ and 
$\xi$ a vector orthogonal to $H(\tilde O )$. We
can find a $\delta>0 $ s.t.\ $x + O\subset\tilde O$, so
\[
f(x)\equiv \langle\xi,U(x)Z\eta\rangle=0,
\] 
for $|x|<\delta $ and $\eta\in H(O)$, where $U(x)$ is the unitary translation by $x$.
By positivity of the energy, $f$ has an analytic continuation on the tube $\RR^4-iV^+$. Since $f(x)=0 $ on an open subset of the boundary, by the Edge of the Wedge theorem $f$ is identically zero. Thus $\xi$ is orthogonal to all translates $H(O + x)$.

We  consider now a wedge $W\supset\tilde O$ and the corresponding boost
one-parameter group $\Lambda_W$. By the KMS property entailed by the Bisognano-Wichmann property, there exists an analytic extension of the function $h$:
\[
h(s) \equiv \langle\xi,U\big(\Lambda_W(2\pi s)\big)Z\eta\rangle\ ,
\]
on the strip $\{z\in\CC: 0<\Im\,z< 1\}$. 
Because $O\Subset \tilde O$, $h(s)$ is zero for small real values of
$s$. Thus the whole extension of $h$ has to be zero. It follows that 
\[
\xi\ \bot\   H(g O)\ , \quad  \forall g\in\cP_+^\uparrow\ .
\]
Now the closed, complex linear span generated by $\big\{H(gO): g\in\Poi\big\}$ 
is a $U$-invariant, non-zero, closed linear subspace of $\H$, that must be equal to $\H$ since $U$ is irreducible. 
Thus $\xi =0$ and $H(\tilde O)$ is cyclic.
\end{proof}
\begin{lemma}\label{lem:timecom} 
Assume that $U$ is a massless, unitary representation of $\tilde\cP_+^\uparrow$ acting covariantly on a twisted-local net of closed, real linear subspaces on double cones. 
Let $O_1, O_2$ be double cones with $O_2$ in the timelike complement of $O_1$, then 
\[
H(O_2)\subset ZH(O_1)'\ .
\]
\end{lemma}
\begin{proof}
Let $O_r$ be the double cone of radius $r>0$ centred at the origin, namely
$O_r$ is the causal envelope of the ball of radius $r$ centred at the origin in the time zero hyperplane.
Consider the two point function 
\[
f(x)=\langle\xi,U(x)Z\eta\rangle,\qquad \xi,\eta\in H(O_r)\ .
\]
Then $\Box f=0$, namely $f$ is a solution of the wave equation, since the Fourier transform of $f$  (w.r.t.\ the Minkowski metric)
is a measure with support in $\partial V_+$. In particular $\Box\Im
f=0$. Now $\Im f(x)=0$ if $x\in O'_{2r}$, because $O_r + x\subset
O'_r$. Thus, by the Huygens principle for solutions of the wave equations,  also $\Im f(x)=0$ if $x$ belongs to the timelike complement of $O_{2r}$. 
Thus $H(O_r + x)\subset ZH(O_r)'$ for such $x$, namely for $x$ such that $O_r+x$ is contained in the timelike complement of $O_r$.
This entails the thesis as $r>0$ is arbitrary.
\end{proof}
\begin{proposition}\label{prop:dil}
Let $U$ be a massless representation of $\tilde\cP_+^\uparrow$, acting covariantly on a net $H$ of standard subspaces on wedges satisfying properties 1--5.
If $H(O)$ is cyclic for some double cone $O$, then $U$ is dilation covariant.

If $U$ is irreducible, the same conclusion holds by assuming that $H(O)\neq \{0\}$ for some double cone $O$.
\end{proposition}
\begin{proof}
Let $H(V_+)$ be the closed, real linear subspace generated by $H(O)$ as $O$ runs in the double cones contained in $V_+$, and similarly for $H(V_-)$.  $H(V_+)$ (and $H(V_-)$) is cyclic as it contains a cyclic real linear subspace $H(O)$ by assumptions (if $H(O)$ is cyclic, all its translated are cyclic).
Since $H(V_+)\subset ZH(V_-)'$ by Proposition \ref{prop:dil}, $H(V_+)$ and $H(V_-)$ are also separating,  hence standard subspaces. Set
\[
D(2\pi t)=\Delta_{H(V_+)}^{-it}\ ,\quad t\in\RR\ .
\]
Then, by Lemma \ref{inc2}, $D(t)$ commutes with $U(g)$ if $g$ is in the Lorentz group, because $gV_+ = V_+$, so $U(g)H(V_+)= H(V_+)$. 
 
Thanks to positivity of the energy, the one-particle version of Borchers'
theorem (Thm.\ \ref{Borch}) applies to all one-parameter groups of
timelike translations. Since the latter generate all translations,
we conclude that $D(s)$ scales the translations:
\[
D(s)U(x)D(-s) = U(e^s x)\ ,\quad s\in\RR\ ,
\]
if $x$ is in the translation group. Thus $U$ is dilation covariant, with dilation unitaries $D(t)$.

The statement for $U$ irreducible then follows immediately by Prop.\
\ref{prop:cicl}. 
\end{proof}
\section{Infinite spin states are not localized in bounded regions}
We give here our main result.
\begin{theorem}\label{teo:noloc} 
Let $U$ be an irreducible unitary, positive energy, massless, infinite spin representation of $\tPoi$ on a Hilbert space $\H$, and  
$H: \W\ni W\longmapsto H(W)\subset \H$ a $U$-covariant net of standard subspaces satisfying properties 1--5.  
Then
\begin{equation}H(O) \equiv \bigcap_{\W\ni W\supset O}H(W)=\{0\}\ ,
\end{equation}
for every double cone $O\in\O$.
\end{theorem}
\begin{proof}
If $H(O)\neq \{0\}$ for some double cone $O$, then by Proposition \ref{prop:dil} $U$ must be dilation covariant, which is not possible by Proposition \ref{dil-inv}.
\end{proof}
The consequences of this theorem in Quantum Field Theory will be discussed in Section \ref{sec:QFT}.
\section{A counter-example}
\label{ce}
In this section, 
we are going to see how dilation covariance and the
double cone Reeh-Schlieder property for infinite spin (reducible) representations may both hold if the Bisognano-Wich\-mann property fails. We shall indeed show that a multiple of the direct integral
\[
\int_{\RR_+}^\oplus U_\kappa d\kappa
\]
over all irreducible representations $U_\kappa$ of $\Poi$ of infinite
spin $\kappa$ is dilation covariant and admits a local covariant net
of standard subspaces, cyclic on double cones. Similar examples were
put forward in \cite{OT,Y69}. 

For the sake of the example, it is sufficient to consider
representations $V$ of $\SLC$ that factor through $\LL_+^\uparrow$,
i.e., $V(1)=V(-1)$. Namely, $V$ is a true representation of
$\LL_+^\uparrow$. Since the choice of the pre-image of the covering map
$\sigma$ does not matter in true representations, we shall identify
$A\in\SLC$ with $\sigma(A)\in\LL_+^\uparrow$ in this section, and
again suppress the corresponding label $\varepsilon=0$.

The subgroup $\tilde E(2)\subset\SLC$, the pre-image of $E(2)$ through $\sigma$, is given by \eqref{E2}.

Let $U_0$ be the unitary, massless, zero helicity, representation of the Poincar\'e group  and $V$ a real unitary representation  $\LL_+^\uparrow$ on the Hilbert spaces $\H$ and $\K$ respectively. 
With $J$ an anti-unitary involution on $\K$ commuting with $V,$ the vectors fixed by $J$ form a standard subspace $K$ of $\K$ and $V(\LL_+^\uparrow)K=K$, $J_{K} = J$, $\Delta_{K} = \eins$.
In particular the constant net of standard subspaces $K(W) \equiv K$ is $V$-covariant. 

We consider $V$ as a representations of $\Poi$ where the translation group acting identically. 

Consider the following net of standard subspaces of $\K\otimes\H$
\[
H_I :\W\ni W\longmapsto H_I(W)\equiv K\otimes H(W)\subset \K\otimes\H
\]
 where $H\equiv H_{U_0}$ is the canonical net associated with $U_0$. 
 There are two unitary representations of the $\cP_+^\uparrow$ on $\K\otimes\H$:
\[
U_V \equiv V\otimes U_0 
\] 
and 
\[
U_I \equiv I\otimes U_0 \ ,
\] 
where $I$ is the identity representation of $\Poi$ on $\K$. Clearly $U_V$ and $U_I$ are massless representations, as the energy-momentum spectrum is that of $U_0$.

$H_I$ is the canonical net associated with $U_I$. The net $H_I$ is  both $U_V$-covariant and $U_I$-covariant.
Only $U_I$ satisfies the Bisognano-Wichmann property as, by Lemma \ref{lem:b}, the modular operator of $K\otimes H(W)$ is $\eins\otimes\Delta_{H(W)}$. Then by Lemma \ref{lem:2}
\[
H_I(O)=\bigcap_{W\supset O} H_I(W)=K\otimes\bigcap_{W\supset O}H(W)
\]
is cyclic, since $\bigcap_{W\supset O}H(W)$ is cyclic in $\H$. 

So we have shown the following.
\begin{proposition}
The net $H_I$ of standard subspaces is local, $U_V$-covariant, and cyclic on double cones. $U_V$ decomposes into a direct integral of infinite spin representation. $U_V$ does not satisfy the Bisognano-Wichmann property.
\end{proposition}
\noindent
We notice that the canonical net $H_V$ associated with $U_V$ is not
covariant under the representation $U_I$.

We will now show that $U_{V}$ decomposes in a direct integral of infinite spin representations if $V$ does not contain the trivial representation. 

Let $V_+\backslash\{0\}\ni p\mapsto B_p\in \LL_+^\uparrow$ be a continuous map, with $B_p$ a Lorentz transformation mapping $q=(1,0,0,1)$ to $p$.

We can identify as usual the elements of $\H$ with $L^2$-functions on $\partial V_+\setminus\{0\}$ w.r.t.\ the Lorentz invariant measure, thus elements of $\K\otimes H$ with $\K_0$-valued $L^2$-functions.

The following unitary operator
\[
\K\otimes\H\ni\big(p\mapsto \phi(p)\big)\longmapsto \big(p\mapsto V(B_p^{-1})\phi(p)\big)\in\K\otimes\H\]
intertwines $U_{V}$ with the representation $U'_{V}$ given by
\begin{equation}\label{U'}
\big(U'_{V}(a,A)\phi\big)(p)=e^{ia\cdot p}V(B_p^{-1}AB_{A^{-1}p}) \phi(A^{-1}p),\qquad \phi\in\H .
\end{equation}
Since $B_p^{-1}AB_{A^{-1}p}\in\Stab_q=E(2)$ we may consider the irreducible disintegration of $V|_{E(2)}$, then $U'_V$, thus $U_V$, will accordingly disintegrate.

Since $\SLC$ is a simple, connected, non-compact Lie group with
finite centre, the vanishing of the matrix coefficients theorem by
Howe-Moore \cite{Z} ensures that
$
\lim_{g\to \infty}\langle\xi,V(g)\eta\rangle=0$, for all  $\xi,\eta\in\K$,
if $V$ does not contain the identity representation. 
\begin{lemma}\label{decomp}
Let $V$ be a unitary representation of $\LL_+^\uparrow$ not containing the identity representation. Then $V|_{E(2)}$ is a multiple of $\int_{\RR_+}^\oplus V_{\kappa} d\kappa$, where $V_{\kappa}$ is the unitary irreducible representation of ${E}(2)$ with radius $\kappa$.

\end{lemma}

\begin{proof}
By the vanishing of the matrix coefficients theorem, there is no non-zero vector fixed by $V
\cdot\tau$, thus no radius zero representation appears in the irreducible direct integral decomposition of $V|_{E(2)}$, namely $V|_{E(2)}=\int_\RR m(\kappa)V_{\kappa} d\mu(\kappa)$, where $m(\kappa)$ is the multiplicity function and $\mu$ is a Borel measure on $\RR_+$.

The one-parameter subgroup $\a$ of $\SLC$ given in \eqref{E2} acts as dilation on the translations $\tau$, eq.\ \eqref{resc},
thus 
\begin{equation}\label{intV}
V|_{E(2)}=
\int^\oplus_\RR m(\kappa)V_{\kappa} d\mu(\kappa) 
=\int^\oplus_\RR m(\kappa)V_{e^t\kappa} d\mu(\kappa)
=\int^\oplus_\RR m(e^{-t}\kappa)V_{\kappa} d\mu_t(\kappa)
\end{equation}
where $\mu_t(\kappa) \equiv \mu(e^{-t}\kappa)$, and this implies that $\mu_t$ is equivalent to $\mu$ (thus $\mu$ is equivalent to the Lebesgue measure) and $m$ constant $\mu$-almost everywhere.
\end{proof}
The following Proposition is a consequence of the above Lemma.
\begin{proposition}
$
U_{V}$ is a multiple of $\int_{\RR_+}^\oplus U_{\kappa}d\kappa
$,
where $U_{\kappa}$ is the infinite spin, radius $\kappa$ representation of $\Poi$.
\end{proposition}
\begin{proof} One considers the disintegration of $V|_{E(2)}$ obtained in Lemma \ref{decomp} and  concludes the thesis by formula \eqref{U'}. 
\end{proof}
\section{Extensions to spacetime dimension $s\geq 2$}
In this section we are going to extend 
Propositions \ref{dil-inv} and
\ref{prop:dil}, and hence also Theorem \ref{teo:noloc},
in any spacetime dimensions $s \geq 2$.
\subsection{Dilation covariance}
\label{dilcov}
We begin by discussing the dilation covariance property.

The proper Lorentz group is $\LL_+\equiv\LL_+(s)=\SO(1,s)$, i.e., the group of $d\times d$ real matrices $A$ preserving  the Minkowski metric $\langle1,-1,\ldots,-1\rangle$. $\LL_+$ has two connected components and we denote by $\LL_+^\uparrow$ the connected component of the identity.

$\LL_+^\uparrow$ is not simply connected when $s>1$. Any element in $\LL_+^\uparrow$ is the product of a rotation and a boost, so $\LL_+^\uparrow$ is homotopy equivalent to $\SO(s)$, whose first  homotopy group is $\Z_2$ if $s>2$ and $\Z$ if $s=2$ (see \cite{LM}). Therefore the universal covering 
$\widetilde\LL_+^\uparrow$ of $\LL_+^\uparrow$ is a double covering for $s>2$, whereas it is an infinite sheet covering if $s=2$.  We shall thus treat the case $s=2$ separately.

The proper orthochronous Poincar\'e group $\cP_+^\uparrow\equiv\cP_+^\uparrow(s)$ is the semi-direct product of $\cP_+^\uparrow\equiv\RR^{s+1}\rtimes \LL_+^\uparrow$, with the natural action of $\LL_+^\uparrow$ on $\RR^{s+1}$.

We shall consider unitary representations of the universal covering group $\tPoi =\RR^{s+1} \rtimes \widetilde\LL_+^\uparrow$, as they correspond to the projective unitary, positive energy representations of $\cP_+^\uparrow$. 

We are interested here in an irreducible, positive energy, massless representation $U$ of $\tPoi$. 
We choose and fix the point $q\equiv q_s=(1,0,\ldots,0,1)$ in the Lorentz orbit $\partial V_+\backslash\{0\}$.
If $U$ is non-trivial, then $U$  is associated with a unitary, irreducible representation of the little group  of $q$, by inducing representations as in Sect. \ref{masslessrep}. 

The little group of $q$, namely the stabiliser subgroup of $\widetilde\LL_+^\uparrow$ for the action of $\LL_+^\uparrow$ on $\RR^{s+1}$, 
is isomorphic to $\tilde E(s-1)$, the double cover of the Euclidean group $E(s-1)$ on $\RR^{s-1}$, $s>2$, i.e., $E(s-1)$ is the semi-direct product $\RR^{s-1}\rtimes \SO(s-1)$. If $s=2$, the little group is the abelian group $\RR$. We now assume $s>2$, afterwords we shall indicate the modifications in the $s=2$ case.

Every unitary representation $V$ of $\tilde E(s-1)=\RR^{s-1} \rtimes
\widetilde\SO(s-1)$ is now induced by a unitary representation of the
stabiliser of a point in $\RR^{s-1}$  (for the adjoint action of $\tilde E(s-1)$). Points in the same orbit give equivalent representations.
The orbits in $\RR^{s-1}$ under the natural $\SO(s-1)$
action are spheres of radius $\kappa\geq 0$. Such radii define
inequivalent classes of unitary representations. As in the $3+1$-dimensional case, there are two cases:
\begin{itemize}\itemsep0mm
\item the restriction $V|_{\RR^{s-1}}$ is trivial ($\kappa=0$); 
\item the restriction $V|_{\RR^{s-1}}$ is non-trivial ($\kappa>0$).
\end{itemize}
If $U$ is associated, by induction, with a representation $V$ of the little group $\tilde E(s-1)$ with $\kappa = 0$ we say that $U$ has \emph{finite helicity}, in the case $\kappa > 0$ we say that $U$ has {\emph{infinite spin}.

With $V$ an irreducible representation of $\tilde E(s-1)$ of radius $\kappa >0$ as above, $s>2$, the joint spectrum of the $\tilde E(s-1)$-translation generators is the sphere in $\RR^{s-1}$ of radius $\kappa$. Therefore 
\begin{equation}\label{spectrum}
{\rm spec}\,(iV(X)) = [-\kappa,\kappa]
\end{equation}
where $X$ is any generator of the $E(s-1)$-translations and $V(X)$ the corresponding translation generator in the representation $V$.
\smallskip

\noindent
We show now that infinite spin representations are not dilation covariant:
\begin{proposition}\label{prop:s-nodil}
Let $U$ be an irreducible, positive energy, unitary representation of $\tPoi(s)$, $s\geq 2$. Then $U$ is dilation covariant iff $U$ is massless with finite spin.
\end{proposition} 
\begin{proof}
We have seen in Proposition \ref{dil-inv} in the case $s=3$ that 
\begin{equation}\label{dilat}
\beta_t(X)=e^{-t}X
\end{equation}
where  $\beta_t$, the automorphisms of $\Poi$ defined in \eqref{beta}, here acting on the Lie algebra of $\Poi$, and $X$ is a translation generator on the Lie algebra $\mathfrak{lie}(E(2))$.

Now assume $s\geq 3$. 
The inclusion ${\Poi}(3)\subset{\Poi}(s)$ restricts to an inclusion $E(2)\subset E(s-1)$ hence we have  inclusions of Lie algebras
$\mathfrak{lie}(E(2))\subset\mathfrak{lie}(E(s-1))\subset\mathfrak{lie}(\Poi)$.
 
We consider the automorphisms $\beta_t$ of $\tilde\cP_+^\uparrow(s)$ analogously defined w.r.t.\ the zero and  $s$ coordinates (the natural extension of $\beta_t$ from ${\tPoi}(3)$ to ${\tPoi}(s)$, we keep the same notation). 

Let now $U$ be an irreducible, positive energy, massless, unitary representation $U$ of $\tPoi(s)$ with infinite spin $\kappa>0$. Then $U$ is associated as above by induction with an irreducible representation $V$ of the little group $\tilde E(s-1)$ of radius $\kappa$. As in Proposition \ref{dil-inv} we have to show that $V\cdot\beta_t |_{\tilde E(s-1)}$ is a representation of radius $e^{-t}\kappa$.

Indeed, due to the relation \eqref{dilat}, with $X\in\mathfrak{lie}(E(2))\subset\mathfrak{lie}(E(s-1))$ we have
\begin{equation}
{\rm spec}\, (iV(X)) = [-e^{-t}\kappa,e^{-t}\kappa]
\end{equation}
so $U$ is not dilation covariant by the above comment.

An analogous discussion shows that finite helicity representations are dilation covariant.

The case $s=2$ is discussed here below.
\end{proof}
{\textit{Case $s=2$}}. In $2+1$ spacetime dimensions, the Lorentz group $\LL_+(2)$ is isomorphic to $\SL2/\{1,-1\}$. The little group of the point $q= (1,0,1)$ is $\RR$, which is simply connected, and lifts uniquely to a one-parameter subgroup of the universal (infinite sheet) cover ${\tilde\LL_+(2)}$. The pre-image of the little group in ${\tilde\LL_+(2)}$ is thus isomorphic to $\RR\times \mathbb Z$, with $\mathbb Z$ the centre of ${\tilde\LL_+(2)}$.

The irreducible representations of the little group $\RR\times \mathbb
Z$ are thus one-dimensional, given by a pair $(\kappa, z)$ where
$\kappa$ belongs to $\RR$ (the dual of $\RR$) and $z\in \mathbb T$ (the dual of $\mathbb Z$).

Denote by $U_{\kappa,z}$ the representation of ${\tPoi}(2)$ associated with the representation $(\kappa, z)$ of the little group.
In analogy with the higher-dimensional case, we say that a unitary representation $U_{\kappa,z}$ of ${\tPoi}(2)$ has ``\emph{infinite spin}'' if $\kappa\neq 0$. Yet, in this case, the name ``infinite spin'' does not refer to any infinite-dimensional representation.

Again, equation (\ref{dilat}) holds, thus the representation $(\kappa,z)$ composed with the restriction of $\beta_t$ to the little group is equal to $(e^{-t}\kappa,z)$. It follows that $U_{\kappa,z}$ is dilation covariant iff $\kappa=0$.
\bigskip

\noindent
We also notice that the conjugate representation of $(\kappa,z)$ is $(-\kappa, \bar z)$, thus $U_{\kappa,z}$ extends to a (anti-)unitary representation of $\cP_+(2)$, iff $\kappa = 0$ and $z=\pm 1$. The other irreducible massless representations of $\cP_+(2)$ are given by $U_{\kappa,z}\oplus U_{-\kappa,\bar z}$, with $\kappa\neq 0$ or $z\neq \pm 1$.

\subsection{Twisted timelike locality}
The second step consists of showing an analogous of Proposition \ref{prop:dil} in any spacetime dimension $s\geq 2$.

We start with a unitary massless representation $U$ acting covariantly on a net on wedges $\W\ni W\longmapsto H(W)\subset\H$ s.t.\ assumptions 1--5 hold. Furthermore, suppose that for some double cone, the subspace $H(O)$, defined as in \eqref{HO2}, is not trivial.
In this setting the proof of Proposition \ref{prop:cicl} straightforwardly extends to every spacetime dimension.

\subsubsection*{Case $s$ odd} When the space dimension $s$ is odd, the
Huygens principle holds and the proof of Proposition \ref{prop:dil} easily extends in this case. 
\subsubsection*{Case $s$ even, $s\geq2$} In this case, timelike commutativity does not hold. Our results hold true, but Lemma \ref{lem:timecom}, necessary to show that $H(V_+)$ is separating, needs a variation.

As is well known, the Huygens principle is not satisfied in odd space dimensions, due to reverberations, yet we show here a version of this principle that holds if $s$ is even.

Let $f$ be a tempered distribution on $\RR^{s+1}$; we define ${\mathfrak h}(f)$
by its Fourier transform
\[
\widehat {{\mathfrak h}(f)}(p) = -i\, {\rm sign}(p_0)\, \hat f(p),
\]
provided this expression is well defined. $\mathfrak h$ is the Hilbert transform with respect to the time variable, thus
\[
\mathfrak h(f)(x) = \frac1\pi \int_{-\infty}^\infty\frac{f(t, x_1,\dots x_s)}{x_0-t}dt
\]
(integral in the principal value sense for a continuous function).
Clearly, if $f_1\in S(\RR^{s+1})$, $f_2\in S'(\RR^{s+1})$, the convolution product satisfies
\[
\mathfrak h(f_1 * f_2) = \mathfrak h(f_1) * f_2 = f_1 * \mathfrak h(f_2) 
\]
If $f$ is a function which is the boundary value of an analytic function on the tube $\RR^{s+1}-iV^+$, $f = \Re f + i \Im f$, then ${\mathfrak h}(\Re f ) = \Im f$; we assume here that $\hat f(0)$ is defined and equal to zero (to rule out the non-zero constants), namely $f$ has zero mean.

We are interested in the case $f$ is a solution of the wave equation $\Box f = 0$, then also $\Box\mathfrak h(f)=0$.
Let $\Delta_+$ be the massless, scalar two-point function, namely the Fourier anti-transform of the Lorentz invariant measure on $\partial V_+\backslash\{0\}$. We have (up to a real proportionality constant),   see e.g.\ \cite{HL}, 
\[
\Delta_+(x) = 1/|x|^{s-1}\quad {\rm if}\quad  x^2 \equiv x_0^2 - x_1^2 -\cdots x^2_s \neq 0\ ,
\] 
where $|x| = \sqrt{-x^2}$ (with opposite square root determination in $V_\pm$)
and
\[
\aligned
&\Delta_+(x) \  \text{real},\quad \Delta_+(x)=\Delta_+(-x)\  , & \quad  x \ {\rm spacelike}\ (x^2<0) \\
&\Delta_+(x) \  \text{imaginary},\quad \Delta_+(x)=-\Delta_+(-x)\  , &  x \ {\rm timelike}\ (x^2>0) .
\endaligned 
\]
The commutator function 
\[
\Delta_0(x) = \Delta_+(x) - \Delta_+(-x)
\]
vanishes for $x$ spacelike, while the function
\[
\Delta'_0(x) = -i\big(\Delta_+(x) + \Delta_+(-x)\big)
\]
vanishes for $x$ timelike. Notice that we have
\[
\Delta'_0 = \mathfrak h(\Delta_0)\ .
\]
\begin{proposition}\label{Huy}
Let $f$ be a bounded continuous function on $\RR^{s+1}$ with $\Box f = 0$, and $O$ a double cone. If $f(x)=0$ for $x$ in the \emph{spacelike} complement of $O$, then $\mathfrak h(f)(x)=0$ for $x$ in the \emph{timelike} complement of $O$.
\end{proposition}
\begin{proof}
Let $h$ be a smooth function with supp$(h)\subset O$. Then $f\equiv h * \Delta_0$ satisfies $\Box f = 0$
and $f(x) = 0$ if $x\in O'$. Moreover
\[
\mathfrak h(f) = \mathfrak h (h * \Delta_0) = h * \mathfrak h(\Delta_0) = h * \Delta'_0
\]
vanishes on the timelike complement of $O$.

Now any smooth function $f$ with $\Box f=0$ and supp$(f)\subset O$ can be written $f= h * \Delta_0$ as above, hence the proposition holds true for every smooth solution of the wave equation $f$. For a general continuous $f$, one can approximate as usual $f$ by $f * j_\varepsilon$ by a smooth approximate identity $j_\varepsilon$, and get the thesis because $\mathfrak h(f * j_\varepsilon) = \mathfrak h(f )* j_\varepsilon$.
\end{proof}
We are now ready to prove the version of Lemma \ref{lem:timecom} in odd spacetime dimensions.
\begin{lemma}\label{even}
Let $U$ be a massless, unitary representation of the double cover of
the Poincar\'e group ${\Poi}(s)$, $s$ even, $s\geq 2$, on a Hilbert space $\H$. Assume that $H$ is a twisted-local, $U$-covariant net of standard subspaces of $\H$ on wedges. 
Let $O_1 ,O_2\in \O$ with $O_2$ in the timelike complement of $O_1$, then 
\begin{equation}\label{td}
H(O_2)\subset iZH(O_1)'\ .
\end{equation}
\end{lemma}
\begin{proof}
With $O_r$ and 
$
f(x)=\langle\xi,U(x)Z\eta\rangle$, $\xi,\eta\in H(O_r)$
as in the proof of Lemma \ref{lem:timecom}, we have $\Box\Im f=0$ and $\Im f(x)=0$ if $x\in O'_{2r}$. 

Thus, by Proposition \ref{Huy}, $\mathfrak h(\Im(f)) = -\Re(f)$ vanishes in the timelike complement of $O_{2r}$; but
\[
\Re(f)(x) = \Im(if)(x)  =  \Im i\langle\xi,ZU(x)\eta\rangle = \Im \langle\xi,ZU(x)i\eta\rangle
\]
and we get the thesis.
\end{proof}
We may now extend Proposition \ref{prop:dil} in any spacetime dimension. Note that, in the following Prop.\ \ref{prop:s-dil}, the cyclicity assumption for $H(O)$ follows from $H(O)\neq \{0\}$
by Prop.\ \ref{prop:cicl}. 
 \begin{proposition}\label{prop:s-dil}
Let $U$ be a massless representation of $\tPoi$, acting covariantly on a net $H$ of standard subspaces of $\H$, satisfying 1--5, on wedges  on the $s+1$-dimensional Minkowski spacetime, with $s\geq 2$.

If $H(O)$ is cyclic for some double cone $O$, then $U$ is dilation covariant. 

Moreover the dilation one-parameter unitary group $D$ can be chosen canonically, and $D(t)\in U(\tPoi)''$, $t\in\RR$.
\end{proposition}
\begin{proof}
$H(V_+)$, the closed linear span of all spaces $H(O)$ with $O\subset V_+$, is a standard subspace of $\H$ by Lemma \ref{even}, so, by positivity of the energy and Theorem \ref{Borch},
the rescaled modular unitary group $D(t) \equiv\Delta_{H(V_+)}^{ -i\frac{t}{2\pi}}$  implements dilations  on $U$-translations, and commutes with the Lorentz unitaries by Lemma \ref{inc2}. Namely $D$ implements the dilations on $U$.

Now, by Proposition \ref{prop:ciclcone}, $U$ extends to an (anti)-unitary representation $\hat U$ of $\tilde\cP_+$ on $\H$,  $\hat U$ maps the reflection around the edge of $W$ to $J_{H(W)}$, and $H(W)= H_{\hat U}(W)$, the standard subspace associated by $\hat U$ with $W$.


Our choice of $D$ is canonical as it is given by modular unitaries. To
show that $D(t)\in U(\widetilde{\mathcal{P}}_+^\uparrow)''$, notice
that this trivially holds if $U$ is irreducible. Recall now that
finite helicity representations are dilation covariant. Assume first
that $\hat U$ is an irreducible, finite non-zero helicity $h$
representation of $\widetilde{\mathcal{P}}_+$, then $\hat U$
restricts to $U=U_{h}\oplus U_{-h}$ on
$\widetilde{\mathcal{P}}_+^\uparrow$, where $U_h$ 
is the helicity $h$ irreducible representation of
$\widetilde{\mathcal{P}}_+^\uparrow$. There is a unitary
implementation of dilations $T(s)$ which decomposes according 
to $U$. As $T(s)D(-s)\in U(\widetilde{\mathcal{P}}_+^\uparrow)'$ and
$U_{h}$ and $U_{-h}$ are disjoint, also $D(s)$ decomposes according to
$U$, and $D(t)\in U(\widetilde{\mathcal{P}}_+^\uparrow)''$ holds. 

In the general case with $U$ reducible, $U$ extends as above to a
representation $\hat U$ of $\widetilde{\mathcal{P}}_+,$ and so the
net $H_{\hat U}$ disintegrates according to $\hat U$. In
particular, $H(V_+)$ and its modular unitaries disintegrate according
to $\hat U$ and we have $D(t)\in U(\widetilde{\mathcal{P}}_+^\uparrow)''$ as
stated. 
\end{proof}
We note that, by Lemma \ref{inc}, if $s$ is even we have 
\[
H(V_+) = iZH(V_-)'\ .
\]
In particular, if $H$ is local,
we have twisted timelike duality
$H(V_+) = iH(V_-)'$, and if $H$ is purely Femi-local ($Z=-i$) we have timelike duality $H(V_+) = H(V_-)'$. 
\subsection{General result} 
We indicate in this section the modifications that are necessary to extend our results in any spacetime dimension $s+1\geq 3$.

Let $U$ be a unitary, positive energy representation of ${\tPoi}(s)$
on a Hilbert space $\H$. We assume here that a $2\pi$-rotation in
space gives a selfadjoint operator $\Gamma\equiv U(2\pi)$, i.e., the
eigenvalues of $\Gamma$ are $\pm 1$. In other words $U$ is a
representation of the double cover of ${\Poi}(s)$ which coincides with the universal cover ${\tPoi}(s)$ if $s>2$; and $\Gamma$ is the image under $U$ of the non-trivial element in the centre of the double cover.

A $U$-covariant (twisted-local) net of standard subspaces $H$ is defined as a map
\[
\cW\ni W\longmapsto H(W)\subset\H
\]
as in Section \ref{ax}.

Note that the proof Proposition \ref{prop:cicl} does not use the twisted locality property, and is valid also here.
We have:
\begin{theorem}\label{prop:f-dil} 
Let $U$ be a unitary, positive energy representation of the cover of the Poincar\'e group ${\tPoi}(s)$, acting covariantly on a net $H$ of standard subspaces of $\H$ on wedges satisfying 1--5 as above, $s\geq 2$.
\begin{itemize}\itemsep0mm
\item[$(a)$] If $H(O)$ is cyclic for some double cone $O$, then $U$ does not contain an infinite spin sub-repre\-sentation (namely there is no infinite spin fibre in the irreducible direct integral decomposition).
\item[$(b)$] If $U$ is irreducible and $H(O)\neq \{0\}$ for some double cone $O$, then $U$ is not massless with infinite spin. 
\item[$(c)$] If\, $U$ extends to an (anti-)unitary, irreducible representation $\hat U$ of $\tilde\cP_+$ on $\H$ and $H(O)\neq \{0\}$ for some double cone $O$, then $U$ does not contain an infinite spin sub-repre\-sentation.
\end{itemize}
\end{theorem}
\begin{proof}
$(b)$ follows from $(a)$ by Prop.\ \ref{prop:cicl}, so we prove the statement $(a)$.

By restricting to the massless component, we may assume that $U$ is massless.
By Proposition \ref{prop:s-dil}, $U$ is dilation covariant; $U =\int^\oplus_X U_\lambda d\mu(\lambda)$ is the irreducible direct integral decomposition of $U$, by Prop.\ \ref{prop:s-dil} the dilation unitary group $D$ decomposes accordingly, $D =\int^\oplus_X D_\lambda d\mu(\lambda)$, so $U_\lambda$ is dilation covariant for $\mu$-almost all $\lambda$. Thus $U_\lambda$ has not infinite spin by
Proposition \ref{prop:s-nodil}.

$(c)$: Either $U$ is irreducible, and we apply $(b)$, or $U$ is the direct sum of two irreducible, inequivalent representations of $\tPoi$, $U =U_1 \oplus U_2$ on $\H_1\oplus \H_2$. In this case, let $\K\subset\H$ be the complex Hilbert subspaces generated by $H$. 
Then $\K$ is $U$-invariant. If $\K=\H$ we apply $(a)$. Otherwise
$U|_\K = U_1$ (or $U|_\K = U_2$). Then $H(O)$ is cyclic on $\K$ for some double cone as in Prop.\
\ref{prop:cicl}, so $U_1$ extends to an (anti-)unitary representation
of $\tilde\cP_+$ on $\H_1$ \cite{GL1}; thus $\hat U$ is easily seen to
be reducible, contrary to our assumption. Therefore $\K = \H$, and the
conclusion follows from $(a)$. 
\end{proof}

\section{Quantum Field Theory: Nets of von Neumann algebras}
\label{sec:QFT}
In this section the Minkowski spacetime dimension is $s\geq 2$. 

Given a positive energy (anti-)unitary representation of the proper Poincar\'e group $\cP_+$ on a Hilbert space $\H$, the paper \cite{BGL} provides a canonical construction of a $U$-covariant local net of standard subspaces of $\H$ on wedges with the properties 1--5. Similarly, this construction gives a twisted-local canonical $U$ covariant net if one considers a representation $U$ of the the universal cover $\tilde\cP_+$.
The above Theorems \ref{teo:noloc}, \ref{prop:f-dil} apply to this net, hence to the net of von Neumann algebras obtained via second quantisation on the Bose/Fermi Fock space, depending on $U(2\pi)=\pm 1$.
\smallskip

\noindent
A twisted-local, $U$-covariant net of von Neumann algebras on wedges $\F$ is an isotonous map
\[
W\longmapsto\F(W)
\]
that associates a von Neumann algebra $\F(W)$ on a fixed Hilbert space $\H$ with every $W\in\W$, with the following properties:
\begin{itemize}\itemsep0mm
\item {\it Poincar\'e covariance}:  $U(g)\F(W)U(g)^* = \F(gW)$,\ \  $g\in \tPoi$;
\item {\it Vacuum with Reeh-Schlieder property}: there exists a unique (up to a phase) $U$-invariant vector $\Omega\in\H$ and 
$\F(W)$ is cyclic on $\Omega$ for all $W\in\cW$;
\item {\it Bisognano-Wichmann property}: 
\[
\Delta^{it}_{W}=U\big(\Lambda_W(-2\pi t)\big),\quad 
W\in\W \ , 
\]
where $\Delta_W$ is the modular operator of $(\F(W),\Omega)$;
\item {\it Twisted locality}: 
For every wedge $W\in\cW$ we have
 \[
 Z \F(W')Z^*\subset \F(W)'
 \]
where $Z$ is unitary and $\displaystyle{Z=\frac{\eins+i\Gamma}{1+i}}$, $\Gamma = U(2\pi)$ as above.
\end{itemize}
Due to twisted locality, $\Omega$ is indeed also separating for each $\F(W)$, so the modular operators $\Delta_W$ are defined.
\smallskip

\noindent
Given $\F$ as above, we define the von Neumann algebra associated with the region $O$ as
\begin{equation}\label{FO}
\F(O) \equiv \bigcap_{\W\ni W\supset O}
\F(W)\ .
\end{equation}
A twisted-local, $U$ covariant net $O\longmapsto \F(O)$ on double cones is analogously
defined, by requiring the $U$-covariance and the cyclicity of the
algebras $\F(O)$. Then $\F(W)$ is defined by additivity and $W\longmapsto \F(W)$ is a twisted-local, $U$-covariant net on wedges. The von Neumann algebras $\F(O)$ defined by \eqref{FO} are, in general, larger than the original $\F(O)$ (they define the dual net). 

The free Bose (resp.\ Fermi) field net $\F_\pm$ is defined by second quantization on
the symmetric/anti-symmetric Fock space $\mathrm{F}_\pm(\H)$ as
\[
\F_\pm(W) \equiv R_\pm\big(H(W)\big)\ ,\quad W\in\W\ ,
\]
where $H = H_U$ is the canonical net of standard subspaces of the
one-particle Hilbert space $\H$ associated with the unitary representation
$U$ of the cover of Poincar\'e group with $U(2\pi) =\pm \eins$, and
$R_\pm(H(W))$ are defined as follows.

With $\H$ a Hilbert space and $H\subset\H$ a real linear subspace,
$R_\pm(H)$ is the von Neumann algebra on $\mathrm{F}_\pm(\H)$ generated by the CCR/CAR operators:
\begin{equation}\label{Rpm}
R_+(H) \equiv \{\mathrm{w}(\xi): \xi\in H\}'',\quad R_-(H) \equiv \{\Psi(\xi): \xi\in H\}''\ ,
\end{equation}
with $\mathrm{w}(\xi)$ the Weyl unitaries on $\mathrm{F}_+(\H)$ and $\Psi(\xi)$ the Fermi field operators on $\mathrm{F}_-(\H)$. 

Note that, by continuity, 
\[
R_\pm(H) = R_\pm(\bar H)\ .
\]
Moreover the vacuum vector $\Omega$ is cyclic (resp.\ separating) for $R_\pm(H)$ iff $\bar H$ is cyclic (resp.\ separating).

If $H$ is standard, we denote by $S^\pm_H$, $J^\pm_H$, $\Delta^\pm_H$
the Tomita operators associated with $(R_\pm(H),\Omega)$, and by
$\mathit{\Gamma}_\pm(T)$ the Bose/Fermi second quantization of a one-particle
operator $T$ on $\H$, defined by tensor products on $\mathrm{F}_\pm(\H)$. 

This assignment \eqref{Rpm} respects the lattice structure, as originally proven in \cite{A} (Bose case) and \cite{F} (Fermi case). The modular operators were computed in \cite{EO,LRT,F}. For convenience, we state these properties in the following proposition with a sketch of proof.
\begin{proposition}\label{prop:secquant}
Let $H$ and $H_a$ be closed, real linear subspaces of $\H$. We have
\begin{itemize}\itemsep0mm
\item[$(a_{\text{\tiny $+$}})$] $S^+_H = \mathit{\Gamma}_+(S_H)$, \ $J^+_H = \mathit{\Gamma}_+(J_H)$, \ $\Delta^+_H= \mathit{\Gamma}_+(\Delta_H)$, 
\item[$(a_{\text{\tiny $-$}})$] $S^-_H = Z\mathit{\Gamma}_-(iS_H)$, \ $J^-_H = Z\mathit{\Gamma}_-(iJ_H)$, \ $\Delta^-_H= \mathit{\Gamma}_-(\Delta_H)$,
\item[$(b)$] $R_+(H)' = R_+(H')$ and $R_-(H)' = ZR_-(iH')Z^*$,
\item[$(c)$] $R_\pm(\sum_a H_a) = \bigvee_a R_\pm(H_a)$,
\item[$(d)$] $R_\pm(\cap_a H_a) = \bigcap_a R_\pm(H_a)$,
\end{itemize}
where $\bigvee$ denotes the von Neumann algebra generated, $Z =
\eins$ (resp. $Z= -i$) on the $n$-particle subspace, $n$ even (resp.\ odd), and $H$ is standard in 
$(a_{\text{\tiny $\pm$}})$.
\end{proposition}
\begin{proof}
$(a_{\text{\tiny $\pm$}})$ $S^+_H = \mathit{\Gamma}_+(S_H)$ due to the relation 
$S^+_H {\rm w}(\xi)\Omega = {\rm  w}(-\xi)\Omega$ (see \cite{LRT}), 
while $S^-_H = Z\mathit{\Gamma}_-(iS_H)$ due to the relation
$S^-_H \Psi(\xi_1)\Psi(\xi_2)\dots \Psi(\xi_n)\Omega = 
\Psi(\xi_n)\dots \Psi(\xi_2)\Psi(\xi_1)\Omega$,
$\xi\in H$, with $\Omega$ the Fock vacuum vector (see \cite{F}). 
By the uniqueness of the polar decomposition, we then have $J^+_H = \mathit{\Gamma}_+(J_H)$, 
$J^-_H = Z\mathit{\Gamma}_-(iJ_H)$ and $\Delta^\pm_H= \mathit{\Gamma}_\pm(\Delta_H)$. 

$(b)$ Assume first that $H$ is standard. By $(a)$ we have
\begin{gather*}
R_+(H)' = J^+_H R_+(H) J^+_H = R_+(J_H H) = R_+(H') \ , \\
R_-(H)' = J^-_H R_-(H) J^-_H = ZR_-(iJ_H H)Z^* = ZR_-(i H')Z^* \ .
\end{gather*}
Now $(b)$ trivially holds for $H=\H$ or $H=\{0\}$ too. To prove $(b)$ for a general closed real linear subspace $H$ of $\H$, we may decompose $\H$ in the direct sum $\H = \H_{-1}\oplus \H_0 \oplus \H_1$, where $\H_1 \equiv H^\bot$ and $\H_{-1} \equiv H\cap iH$. Then $H$ decomposes as 
$H = H_{-1}\oplus H_0 \oplus H_1$ with $H_{-1} =\H_{-1}$, $H_1 =\{0\}$ and $H_0$ a standard subspace of $\H_0$, and the statement follows at once.

$(c)$ is an immediate consequence of the Weyl relations $\mathrm{w}(\xi + \eta) = e^{-\Im\langle \xi,\eta\rangle}\mathrm{w}(\xi)\mathrm{w}(\eta)$ (Bose case), the real linearity of $\Psi$ (Fermi case), and $(a)$.

$(d)$ now follows easily from $(b)$ and $(c)$.
\end{proof}
We state now the following consequence of Theorems \ref{teo:noloc},
\ref{prop:f-dil} for free fields. 
\begin{corollary} \label{cor:free}
Let $\F_\pm$ be the free Bose/Fermi field net of von Neumann algebras on wedges associated with a positive energy, infinite spin, irreducible unitary Bose/Fermi representation of $\tPoi$ \cite{BGL}. 

Then $\F(C)$ is cyclic on the vacuum vector if $C$ is a spacelike
cone, but $\F(O)= \CC\cdot\eins$ if $O$ is any bounded spacetime region.
\end{corollary}
\begin{proof}
Immediate by Theorem \ref{teo:noloc}, the results in \cite{BGL}, and
the fact (Theorem \ref{prop:secquant}) that the intersection of closed
real linear spaces of the one-particle Hilbert space corresponds to
the intersection of the associated von Neumann algebras on the Fock space:
\[
\bigcap_{W\supset X}\F_\pm(W) \equiv \bigcap_{W\supset X} R_\pm\big(H(W)\big) =
R_\pm\big(\bigcap_{W\supset X} H(W)\big) \ 
\]
for $X=C$ a cone, resp.\ $X=O$ a double cone.
\end{proof}
The following theorem shows why infinite spin representations do not occur in a theory of local observables.

We shall say that a unitary representation $U$ of $\tPoi$ has \emph{infinite spin} if $U$ is a direct integral of irreducible, infinite spin representations. Thus $U$ does not not contain an infinite spin sub-representation iff no infinite spin representation appears in the irreducible direct integral decomposition of $U$ (up to a measure zero set).
\begin{theorem}\label{main}
Let $\F$ be a twisted-local net of von Neumann algebras $\F(O)$ on double cones on a Hilbert space $\H$, covariant w.r.t.\  a unitary representation $U$ of $\tPoi$ with vacuum vector $\Omega\in\H$. As above, we assume the double cone Reeh-Schlieder property and the Bisognano-Wichmann property.

Then $U$ does not contain an infinite spin sub-representation.
\end{theorem}
\begin{proof}
For every wedge $W\in\W$ we set 
\[
H(W)=\overline{\F(W)_{\rm s.a.}\Omega}\ ,
\]
where $\F(W)_{\rm s.a.}$ is the selfadjoint part of $\F(W)$. By assumptions, $H: W\longmapsto H(W)$ is then a twisted-local, $U$-covariant net of standard subspaces of $\H$ satisfying Properties 1--5.

With $O$  a double cone, we have that
\[
H(O) \equiv \bigcap_{\W\ni W\supset O}\!\!\!\! H(W) \supset \F(O)_{\rm s.a.}\Omega
\]
is cyclic. 

We thus infer from Theorem \ref{prop:f-dil} $(a)$ that $U$ does not contain an infinite spin sub-representation.
\end{proof}
We now start with a local net $\A$ of von Neumann algebras on double
cones, with the double cone Reeh-Schlieder property and the Bisognano-Wichmann property.
Let 
\[
{\mathfrak A}\equiv \overline{\bigcup_{O\in\O}\A(O)}
\]
(norm closure) be the quasi-observable $C^*$-algebra.
We shall say that a representation $\pi$ of $\mathfrak A$ is cone
localizable if, for every spacelike cone $C$, $\pi |_{\mathfrak A(C')}$ is unitarily equivalent to ${\rm id}|_{{\mathfrak A}(C')}$, where $\mathfrak A(C)$ is the $C^*$-algebra generated by $\A(O)$ as $O$ runs in the double cones contained in $C$. Similarly $\pi$ is double cone localizable if $\pi|_{{\mathfrak A}(O')}\simeq{\rm id}|_{{\mathfrak A}(O')}$, for all double cones $O$.

A Doplicher-Haag-Roberts (DHR) (resp.\ a Buchholz-Fredenhagen) representation \cite{DHR,BF} is a Poincar\'e covariant representation with positive energy, which is double cone (resp.\ cone) localizable. (Poincar\'e covariance with positive energy follows by general assumptions \cite{GL1}). 
\begin{theorem}
Let $\pi$ be a DHR representation of $\mathfrak A$ with finite statistics \cite{DHR}. 
Then  the unitary representation $U_\pi$ of $\tPoi$ in the
representation $\pi$ does not contain infinite spin sub-representations.
\end{theorem} 
\begin{proof}
By considering the dual net, we can assume Haag duality for double
cones. We consider the Doplicher-Roberts twisted-local field net $\cF$. We have $\A(O)\subset\cF(O)$ and the restriction of the vacuum representation of $\cF$ to $\A$ is the direct sum (with multiplicity) of all DHR representations of $\A$ with finite statistics.

The representation $U_\cF$ of $\tPoi$ restricts accordingly to the
representations of $\A$. Thus we have to show that $U_\cF$ does not
contain an infinite spin sub-representation. This will follow from Theorem \ref{main} once we show the Bisognano-Wichmann property. Now the Bisognano-Wichmann property for $\F$ is a consequence of the Bisognano-Wichmann property for $\A$ as one can identify the Connes-Radon-Nikodym cocycles, see \cite{LKW,I}.
\end{proof}
As a consequence, let  $\pi$ be a Poincar\'e covariant representation of  $\mathfrak A$.
If $\pi$ contains infinite spin particles (i.e., $U_\pi$ contains an infinite spin sub-representation) then:
\[
\text{$\pi$ is localizable in a double cone} \implies \text{$\pi$ has infinite statistics}.
\]
This indicates an intimate relation among infinite spin, infinite statistics and localization in infinitely extended regions.
\section{Final comments}
\subsection{Field algebra structure}
\label{fieldstr}
We now describe the field algebra structure that we obtain starting from the observable algebra and adding all charges with finite statistics, including the ones with infinite spin (space dimension $s>2$).

Let $\A$ be as a above a local net with the double cone Reeh-Schlieder
property and the Bisognano-Wichmann property. Let $\T$ be the family of all irreducible representations, up to unitary equivalence (sectors), of $\mathfrak A$ of Buchholz-Fredenhagen type with finite statistics. 

The Doplicher-Roberts construction \cite{DR} yields a field net $\F$ of von Neumann algebras on a larger Hilbert space with $\F(C)\supset \A(C)$ for every cone $C$, and the identity representation of $\A$ on the Hilbert space of $\F$ decomposes into the direct sum of elements of $\T$, with multiplicity. 

By the spin-statistics theorem \cite{GL}, $\F$ is a twisted-local
net. If infinite spin sectors exist, then by Theorem \ref{main}
$\F(O)$ cannot be cyclic on the vacuum vector if $O$ is a bounded
region. If one restricts $\F$ to the cyclic Hilbert space generated by
$\F(O)$, one gets the field algebra associated with DHR charges. 
We discuss a physical interpretation of this structure in the outlook.

We mention also that, in two space dimensions, cone localizable
representations may have braid group statistics. If we consider only
those ones with Bose or Fermi statistics, then the above field algebra
description still holds (the spin-statistics theorem in $2+1$ dimensions is treated 
in \cite{Lt}). However, with general statistics, no field algebra
exists that describes an analogue of the above picture. 

\subsection{de Sitter spacetime}
If $\A$ is a local net on spacelike cones of the Minkowski spacetime
$\RR^{s+1}$, one can associate a local net $\B$ on double cones of the
$s$-dimensional de Sitter spacetime $dS_s$ (and similarly in the
twisted-local case). As usual, one views $dS_s$ as an hyperboloid of
$\RR^{s+1}$, which is the manifold of spacelike directions of
Minkowski spacetime. With $E$ any region of $dS_s$, one sets
$\B(E)\equiv \A(C_E)$, where $C_E\subset \RR^{s+1}$ is the spacelike cone with apex in the origin spanned by $E$.

This construction has been made in \cite{BGL}. In particular, in the free field case (finite or infinite spin), one gets the canonical modular construction on $dS_s$ associated with the restriction of Poincar\'e unitary representation to the Lorentz subgroup.

We emphasize here that the de Sitter picture is natural in the
presence of infinite spin particles. These particles have no bounded
spacetime localization on the Minkowski spacetime, yet they are
localized in bounded spacetime regions of the de Sitter spacetime.

\subsection{The role of the Bisognano-Wichmann property}
\label{comm:BW}
In this paper, we rely on the Bisognano-Wichmann property as a first principle (cf. \cite{GL1, GL}), so we briefly comment here on its roots. 

The Bisognano-Wichmann property implies the positivity of the energy (see \cite {BGL}), and is slightly stronger than that; it reflects the stability of the vacuum state.
It always holds in a Wightman theory \cite{BW}, including string
localized fields \cite{MSY}. In the local algebra framework one can
find a counter-example (see Sect.\ \ref{ce}), that has however a
pathological nature (with continuous degeneracy) and is built on the non-uniqueness of the covariance unitary representation of the Poincar\'e group: if one chooses the wrong (non-canonical) representation, one obviously violates the Bisognano-Wichmann property. So we may expect the latter to always hold when the Poincar\'e representation is unique, say by assuming the split property. 

In a massive theory, the Bisognano-Wichmann property can be derived by asymptotic completeness \cite{M}. It always holds in the conformal case. It is equivalent to a sub-exponential growth estimate on the energy density levels of localized states for the Rindler Hamiltonian, namely
$(\xi,e^{-2\pi K}\xi)<\infty$ for all vector states $\xi$ localized in
a given cone $C$ contained in a wedge $W$, with $K$ the generator of the
unitary one-parameter group of boosts associated with $W$ \cite{GL2}.

A further argument for the
Bisognano-Wichmann property is its mentioned equivalence with
the Hawking-Unruh effect for Rindler black holes \cite{Sew} (the Hawking temperature is the KMS temperature). An illustration of this fact goes beyond the purpose of this paper and we refer to the book \cite{H} for more insight on this point and related aspects.
\section{Outlook}
\label{Concl}
Infinite spin particle states cannot be 
localized in a bounded spacetime region. This corresponds to the fact that no local
observables exist that can generate these states from the
vacuum. These results, obtained in the present paper in the Operator
Algebraic intrinsic setup \cite{BGL}, extend the no-go theorem on local fields with 
infinite spin obtained previously in the Wightman setting \cite{Y}.
The string-localized free fields constructed in \cite{MSY}, that correspond to and generate the von Neumann algebras in \cite{BGL}, cannot thus be compactly localized.

As described in Sect.\ \ref{fieldstr} our results provide the following picture in a theory of local observables.

A quantum field theory on a Hilbert space
including infinite spin states is described by a net $W\longmapsto \F(W)$ of von Neumann algebras for wedge regions, and the
vacuum vector is cyclic for the von Neumann algebras for spacelike
cones (defined by intersections of wedge algebras), and has the
Bisognano-Wichmann property. The algebras for double cone regions are
non-trivial, forming a covariant subnet $O\longmapsto \A(O)$, but
the vacuum is not cyclic for $\A(O)$. The full Hilbert space therefore
splits into representations of $\A$, with the infinite spin states
absent from the vacuum representation. The representations containing
the infinite spin states are massless sectors of the
Buchholz-Fredenhagen type, i.e., localized in spacelike cones, and the
net $\F$ serves as a field algebra for these sectors.  
This picture complies with the scenario proposed by Schroer \cite{Sch}
with a hindsight on ``dark matter''. 

One may reasonably expect that $\A$ contains local generators of
Poincar\'e transformations, i.e., a stress-energy tensor subnet (which
could couple to gravity). As infinite spin states are localized in
spacelike cones, their Lorentz transforms will be localized in
different cones. Thus, the obstruction against infinite spin states to
be present in the vacuum representation is necessary because they
cannot be Lorentz transformed by local generators. But the
representatives of local generators in a cone-localized representation
may well Lorentz transform infinite spin states present in these sectors. 

By our result (Cor.\ \ref{cor:free}), if $\F$ is the \textit{free} field net
associated with an infinite spin representation, then the subnet $\A$
would be trivial. Hence, the above scenario necessarily requires a
self-interaction of some unknown sort. It is an exciting challenge to
describe such an
interaction, 
and the possible interaction of infinite
spin fields with ``ordinary matter'' fields.

\bigskip

\noindent
{\bf Acknowledgements.} We thank S.~Carpi, D.~Guido, G.~Morsella, and J.~Yngvason  for useful discussions.

\end{document}